\documentclass[a4paper,a4wide,11pt]{amsart}

\usepackage[utf8]{inputenc}     

\usepackage{tikz}
\usepackage{url}
\usepackage{subcaption}
\usepackage{nameref} 
\usepackage{hyperref}

\usepackage{amsfonts}
\usepackage{amsthm}
\usepackage{amssymb}
\usepackage{amsmath}
\usepackage{enumerate}
\usepackage{a4wide}
\usepackage{color}
\usepackage{graphicx}
\usepackage[normalem]{ulem}

\newtheorem{theorem}{Theorem}
\newtheorem{definition}{Definition}
\newtheorem{lemma}{Lemma}
\newtheorem{proposition}{Proposition}

\DeclareUnicodeCharacter{D7}{\times}            
\DeclareUnicodeCharacter{3C4}{\tau}             
\DeclareUnicodeCharacter{2026}{\ldots}          
\DeclareUnicodeCharacter{2113}{\ell}            
\DeclareUnicodeCharacter{2115}{\mathbb{N}}      
\DeclareUnicodeCharacter{2228}{\vee}            
\DeclareUnicodeCharacter{222A}{\cup}            
\DeclareUnicodeCharacter{2261}{\equiv}          
\DeclareUnicodeCharacter{22EF}{\cdots}          
\DeclareUnicodeCharacter{2A7D}{\leqslant}       
\DeclareUnicodeCharacter{2A7E}{\geqslant}       

\newcommand{\FT}[0]{\mathbb{F}_2}
\newcommand{\even}[0]{\operatorname{even}}
\newcommand{\odd}[0]{\operatorname{odd}}
\newcommand{\join}[2]{#1 \vee #2}

\begin{document}
\title{Nested perfect toroidal arrays}

\author{Ver\'onica Becher \and Olivier Carton}
\date{\today}

\begin{abstract}
  We introduce two-dimensional toroidal arrays that are a variant of the de
  Bruijn tori. We call them nested perfect toroidal arrays.  Instead of
  asking that every array of a given size has exactly one occurrence, we
  partition the positions in congruence classes and we ask exactly one
  occurrence in each congruence class.  We also ask that this property
  applies recursively to each of the subarrays.  We give a method to
  construct nested perfect toroidal arrays based on Pascal triangle matrix
  modulo~$2$.  For the two-symbol alphabet, and for $n$ being a power of $2$,
  our method yields $2^{n^2+n-1}$ different nested perfect toroidal arrays
  allocating all the different $n\times n$ arrays in each congruence class
  that arises from taking the line number modulo~$n$ and the column number
  modulo~$n$.
\end{abstract}

\maketitle




\noindent
\textbf{CCS:} Mathematics of computing, Discrete mathematics, Combinatorial problems



{\tableofcontents}

\section{Introduction and statement of results}

A toroidal array of size $n\times n$ is an array where the line numbers are
considered modulo~$n$ and the column numbers are considered modulo~$n$. In
this note we consider  toroidal arrays of symbols in a finite
alphabet.  The problem of constructing toroidal arrays of minimal size
that allocate a given family of smaller arrays goes back to the 1960s, see
for instance~\cite{Gordon,Reed1962}. The constructions of toroidal arrays
where each member of the family occurs exactly once generalize the
classical unidimensional de Bruijn sequences to two dimensional arrays, and
they are known as de Bruijn tori or perfect maps.  
There has been significant effort in solving the problem of determining the existence of these toroidal arrays for
the different subarrays sizes, and in giving
construction methods, for instance the work
of~\cite{Ma1984,Fan1985,sloane,HI,Etzion,Paterson1994,CDG}.
The first achievements focus on arrays of $0$s and~$1$s. Subsequent
work solves the existence problem and the construction problem for
arbitrary alphabets. The ability to recover efficiently any given  subarray receives special attention~\cite{MY}. There is also work on
constructions for three dimensional arrays~\cite{3dim}.

In this note we introduce a variant of the de Bruijn tori. 
We call them nested perfect toroidal arrays.  Instead of
asking that every array of a given size has exactly one occurrence, we
partition the positions in congruence classes and we ask exactly one
occurrence in each congruence class.  We also ask that this property
applies recursively to the subarrays.

Nested perfect toroidal arrays are the two-dimensional version of the work
done by the authors in the unidimensional case, that they called
\emph{nested perfect necklaces}~\cite{BecherCarton2019}. Nested perfect
necklaces are a special case of the \emph{perfect necklaces}~\cite{ABFY}.

Along the sequel we number the lines and columns starting at~$0$ (instead
of starting at~$1$).  For any two-dimensional array toroidal~$A$ the
subarray of size $n × n$ at position $(k,ℓ)$ is the array made of the lines
from $k$ to $k+n-1$ and the columns from $ℓ$ to $ℓ+n-1$ where all these
indices are taken modulo the number of lines and columns of~$A$.  Note that
$(0,0)$ is the position of the upper left corner of each array because
lines and columns are numbered starting from~$0$.

A \emph{modulo} is a pair of positive integers written $(p, q)$.  The
positions of any given array are partitioned into $pq$ residue classes
according to their respective residue classes modulo~$p$ and modulo~$q$:
Thus, the modulo $(1, 1)$ yields a single class containing all the
positions, and the modulo $(2, 2)$ partitions the positions in four
classes.

An array of size $n × n$ \emph{occurs} at position $(k,ℓ)$ in an array
if it is equal to the subarray  of size $n × n$ at position $(k,ℓ)$.

\begin{definition}[Perfect toroidal array]
  A toroidal array~$A$ is \emph{perfect} for window size $s × t$ and modulo
  $(p, q)$, abbreviated $(s, t, p, q)$-perfect, if each $s × t$ array has
  exactly one occurrence in~$A$ in each residue class modulo $(p, q)$.
\end{definition}

Then, in an $(s,t,1,1)$-perfect toroidal array each $s ×t$ array has exactly
one occurrence.  Figure~\ref{fig:perfect} gives an example of a
$(2,2,1,1)$-perfect toroidal array and two $(2,2,2,2)$-perfect toroidal
arrays.  In the leftmost, each $2 × 2$ array occurs exactly once.  In the
other two, each $2 × 2$ array occurs exactly four times, having exactly one
occurrence in each residue class modulo $(2, 2)$.

\begin{definition}[Aligned subarray]
Given an array $A$, a subarray of size $k × ℓ$ is \emph{aligned} if its position $(i,j)$  in~$A$ satisfies that $k$ divides $i$ and $ℓ$ 
divides~$j$.  
\end{definition}

\begin{definition}[Subdivision]
  If both $k$ and $\ell$ divides $n$,
a \emph{$k × ℓ$-subdivision} of an array of size $n × n$  yields  $k\ell$ aligned subarrays, each of  size $n/k × n/ℓ$.
\end{definition}
\begin{definition}[Nested perfect toroidal array]\label{npa}
  Assume a $b$-symbol alphabet.  A perfect toroidal array~$A$ is
  \emph{nested} for window size $s × t$ and modulo $(p,q)$, abbreviated
  \emph{nested $(s, t, p, q)$-perfect}, if for each $k=0,\ldots, s-1$, 
   each aligned subarray of the 
  $b^{kt/2} × b^{kt/2}$ subdivision of~$A$ is $(s-k,t,p,q)$-perfect.
\end{definition}

Notice that $(f,g,p,q)$-perfect implies   $(f',g',p,q)$-perfect 
for  $1\leq f'\leq f$ and   $1\leq g'\leq g$. The reverse implication may not be true.
Definition~\ref{npa} asks that the subdivisions yields  $(s-k,t,p,q)$-perfect subarrays instead of   
$(s-k,t-k,p,q)$-perfect, as it could be expected.
Our motivation is that we have a construction method that ensures the stronger property.

There are other natural options for the definition  of a nested perfect array.
For  non square arrays such definitions are not equivalent to each other, but they all coincide for square arrays.
In this work we are interested in  the square case.

Consider a  $(n, n, n, n)$-perfect toroidal array in the $b$-symbol
alphabet.
Its size  is $nb^{n^2/2} × nb^{n^2/2}$ for $n$~even.   
For $k=0,\ldots,n-1$, 
each part of its  $b^{kn/2} × b^{kn/2}$ subdivision
has size $nb^{n(n-k)/2} × nb^{n(n-k)/2}$.

For square arrays the  definition of  nested perfect toroidal arrays can be rephrased as follows.

\begin{definition}[Square nested perfect toroidal array]
  Assume a $b$-symbol alphabet.  An $(n,n,n,n)$-perfect toroidal array
  is \emph{nested} if, for $k=1,…, n$, each aligned subarray of size
  $nb^{nk/2} × nb^{nk/2}$ is $(k,n,n,n)$-perfect.
\end{definition}

For $b=2$ and   $n = 4$: an array~$A$ of size $1024 \times 1024$  is a nested $(4,4,4,4)$-perfect
  toroidal array if
  \begin{itemize}
  \item $A$ is a  $(4,4,4,4)$-perfect toroidal array;
  \item each $256 \times 256$ aligned subarray is a $(3,4,4,4)$-perfect 
    toroidal array;
  \item each $64 \times 64$ aligned subarray is a $(2,4,4,4)$-perfect 
    toroidal array;
  \item each $16 \times 16$ aligned subarray is a $(1,4,4,4)$-perfect 
    toroidal array.
  \end{itemize}

\begin{figure}[htb]
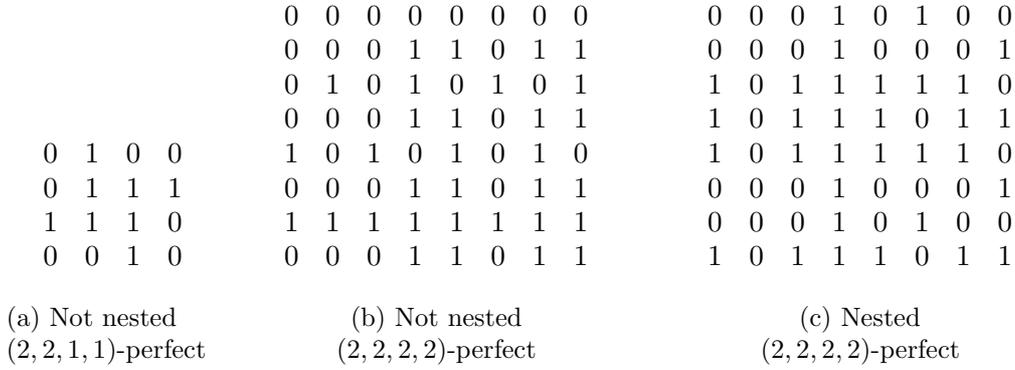

  \begin{center}
    \begin{subfigure}{0.18\textwidth}
      \begin{displaymath}
        \begin{array}{cccc}
        \\\\\\\\
          0 & 1 & 0 & 0 \\
          0 & 1 & 1 & 1 \\
          1 & 1 & 1 & 0 \\
          0 & 0 & 1 & 0 
        \end{array}
      \end{displaymath}
      \caption{Not nested \\ $(2,2,1,1)$-perfect}
      \label{fig:p1part1}
    \end{subfigure} 
    \begin{subfigure}{0.35\textwidth}
      \begin{displaymath}
        \begin{array}{cccccccc}
          0 & 0 & 0 & 0 & 0 & 0 & 0 & 0 \\
          0 & 0 & 0 & 1 & 1 & 0 & 1 & 1 \\
          0 & 1 & 0 & 1 & 0 & 1 & 0 & 1 \\
          0 & 0 & 0 & 1 & 1 & 0 & 1 & 1 \\
          1 & 0 & 1 & 0 & 1 & 0 & 1 & 0 \\
          0 & 0 & 0 & 1 & 1 & 0 & 1 & 1 \\
          1 & 1 & 1 & 1 & 1 & 1 & 1 & 1 \\
          0 & 0 & 0 &  1 & 1 & 0 &  1 & 1
        \end{array}
      \end{displaymath}
      \caption{Not nested \\ $(2,2,2,2)$-perfect}
      \label{fig:p1part1}
    \end{subfigure} 
    \begin{subfigure}{0.35\textwidth}
      \begin{displaymath}
        \begin{array}{cccccccc}
          0 & 0 & 0 & 1 & 0 & 1 & 0 & 0 \\
          0 & 0 & 0 & 1 & 0 & 0 & 0 & 1 \\
          1 & 0 & 1 & 1 & 1 & 1 & 1 & 0 \\
          1 & 0 & 1 & 1 & 1 & 0 & 1 & 1 \\
          1 & 0 & 1 & 1 & 1 & 1 & 1 & 0 \\
          0 & 0 & 0 & 1 & 0 & 0 & 0 & 1 \\
          0 & 0 & 0 & 1 & 0 & 1 & 0 & 0 \\
          1 & 0 & 1 & 1 & 1 & 0 & 1 & 1
        \end{array}
      \end{displaymath}
      \caption{Nested \\ $(2,2,2,2)$-perfect}
      \label{fig:p1part1}
    \end{subfigure} 
  \end{center}
\caption{Examples of perfect toroidal arrays }
\label{fig:perfect}
\end{figure}

\begin{figure}[htb]
\begin{displaymath}
  \begin{array}{c|c|c}
    \text{Array size for nested $(n,n,n,n)$-perfect, $n=2^d$ } & \text{Window size} & \text{Modulo} \\ \hline
    (b^{n × n/2} × n) × (b^{n × n/2} × n) & n × n &(n, n) \\
    (b^{n/2 × n/2} × n) × (b^{n/2 × n/2} × n) & (n/2)\times  n
                                                  & (n, n) \\
  (b^{n/2^{2} × n/2} × n) × (b^{n/2^{2} × n/2} × n) & (n/2^{2}) × n 
                                                  & (n, n) \\
    \vdots & \vdots & \vdots \\
    (b^{1 × n/2} × n) × (b^{(1 × n/2} × n) & 1 × n & (n, n)
  \end{array}
\end{displaymath}
    \caption{Subdivisions sizes and window sizes}
    \label{fig:my_label}
\end{figure}

The leftmost two toroidal arrays in Figure~\ref{fig:perfect} above are not  nested perfect:
The leftmost one is not nested 
because its left upper $2 × 2$ subarray has $2$ occurrences of the $1 × 2$
array $[0, 1]$ but no occurrence of the $1 × 2$ array $[0, 0]$.  The middle array in Figure~\ref{fig:perfect} is not a nested $(2,2,2,2)$-perfect array because its left upper $4 \times 4$ subarray contains more $0$s than $1$s. The rightmost toroidal array is a nested $(2,2,2,2)$-perfect toroidal array.
\medskip

The following  is the main result in the present note.  It states the
existence of nested perfect toroidal arrays of $0$s and $1$s when all parameters are equal to
the same power of~$2$.

\begin{theorem} \label{thm:existence}
  For every integer  $n \geqslant 2$ that is a power of $2$, there exist nested
  $(n,n,n,n)$-perfect  toroidal arrays of $0$ and $1$s.
\end{theorem}

Our construction method does not yield just one instance, but many. The next result, Theorem~\ref{thm:constructions},
gives the exact number of different instances obtainable with our method.

\begin{theorem}\label{thm:constructions}
There is a construction method
that, 
for each integer $n$ that is a power of $2$,
yields  $2^{n^2+n-1}$ different nested $(n,n,n,n)$-perfect  toroidal arrays of $0$ and $1$s.
  \end{theorem}

We do not know if there are more.

\section{Proof of Theorem~\ref{thm:existence}}

We assume that the alphabet is the two element field $\FT = \{0, 1\}$.  
We shall use matrices of elements in $\FT = \{0, 1\}$,  do matrix multiplication and matrix summation. 
The component-wise sum of elements in~$\FT$ is denoted by the symbol~$\oplus$.
 Matrices are named with the letters $M, N, P, Q$ with sub-indices and super-indices. When we depict a matrix we draw the surrounding black parenthesis outside.
The outcome of the construction in each case is a  toroidal array of $0s$ and $1$s  obtained by tiling with the above mentioned 
matrices. 
We name the arrays  with letters $A,B,C$ and when we draw them we do not put the surrounding  black parenthesis outside.

We start by  defining the following  family of matrices. 
\begin{definition} \label{def:Md}
We give an inductive definition of the  matrix~$M_d$ of elements in  $\FT$, of size $2^d × 2^d$, for each $d \ge 0$ by

\begin{displaymath}
  M_0 = (1)
  \quad\text{and}\quad
  M_{d+1} = \begin{pmatrix} M_d & M_d \\ 0  & M_d \end{pmatrix}.
\end{displaymath}
\end{definition}
Thus, the matrices $M_1$ and $M_2$ are
\begin{displaymath}
  M_1 = \begin{pmatrix}
    1 & 1 \\
    0 & 1
  \end{pmatrix} 
  \quad\text{and}\quad
  M_2 = \begin{pmatrix}
    1 & 1 & 1 & 1 \\
    0 & 1 & 0 & 1 \\
    0 & 0 & 1 & 1 \\
    0 & 0 & 0 & 1
  \end{pmatrix}.
\end{displaymath}
\medskip

For every $d$, the matrix~$M_d$ is upper triangular, that is
$(M_d)_{i,j} = 0$ for $0 ⩽ j < i  < 2^d$.  The following lemma states that
the upper part of the matrix~$M_d$ is the beginning of the Pascal triangle
modulo~$2$ also known as the Sierpi\'nski triangle.  The proof is a simple 
induction on~$d$ and can be found in~\cite[Lemma 3]{BecherCarton2019}.

The matrix~$M_d$ is almost the one
used by M. Levin in~\cite[Theorem 2]{Levin1999} because we have reversed the order of the columns.  
The Pascal triangle matrix has been previously used by H. Fauré~\cite{Faure1982}  for the construction of uniformly distributed sequences of real numbers.

\begin{lemma}[\protect{\cite[Lemma 3]{BecherCarton2019}}] \label{lem:pascal}
  For all integers $d,i,j$ such that $d \geqslant 0$ and $0 < i < 2^d$ and $0
  \leqslant j < 2^d-1$, $(M_d)_{i,j} = (M_d)_{i-1,j} \oplus (M_d)_{i,j+1}$.
\end{lemma}
\begin{proof}
To facilitate the review we include the proof already given in \cite[Lemma 3]{BecherCarton2019}.
  The proof is carried out by induction on~$d$.  For $d = 0$, the result is
  trivially true because there are no such $i$ and~$j$.  For $d = 1$, the
  result trivially holds.  Suppose that the result holds for $M_d$ and let
  $i,j$ be integers such that $0 < i < 2^{d+1}$ and $0 \le j < 2^{d+1}-1$.
  If $i \neq 2^d$ and $j \neq 2^d-1$, the three entries $(M_{d+1})_{i,j}$,
  $(M_{d+1})_{i-1,j}$ and $(M_{d+1})_{i,j+1}$ lie in the same quarter of
  the matrix~$M_{d+1}$ and the result follows from the inductive
  hypothesis.  Otherwise, the result is a consequences  the following facts. For
  each integer~$d \ge 1$, the entry $(M_d)_{i,j}$ is equal to~$1$ if either
  $i = 0$ or $j = 2^d-1$ (first row and last column) and it is equal to~$0$
  if $i = 2^d-1$ or $j = 0$ (last row and first column) and $(M_d)_{0,0} =
  (M_d)_{2^d-1,2^d-1} = 1$ (intersection of the two previous cases).  These
  facts are easily proved by induction on~$d$.
\end{proof}

Bacher and Chapman~\cite[Theorems 1 and 3]{Bacher-Chapman} proved that 
for every non negative integer $d$ and every
integer  $k$ such that $1\leq k\leq 2^d$, 
  every  $k\times k$  submatrix of $M_d$
  given by    $k$  consecutive rows and 
  and the last $k$ columns, or by 
  the  top  $k$  consecutive rows and 
  any consecutive  $k$ columns,  is invertible.
  In case  $k = 2^d$ this says that  the whole matrix $M_d$ is invertible.
A proof of this result in more general form appears in   Lemmas~\ref{lem:invert-right-sm} 
and~\ref{lem:invert-top-sm} in the next section.

Now we define an enumeration of all $n × n$ matrices over~$\FT$.  Suppose that the integer~$n$ is fixed.  Let $N_0,\ldots , N_{2^{n^2}-1}$ be the
enumeration of these matrices defined as follows.  Informally,  for $0 ⩽ k < 2^{n^2}$, the matrix
$N_k$  is filled by the digits of the binary expansion
of~$k$: the most $n$ significant binary digits are put in the first line,
the following $n$ digits are put in the second line and so on, until the
last line.

\begin{definition}[Matrices enumeration]\label{def:enum}
 Fix  a positive integer  integer~$n$.
 We define an enumeration $N_0,\ldots , N_{2^{n^2}-1}$ of all the  $n × n$ matrices over~$\FT$.   
 Let $k$ be a non-negative integer and 
 let  $a_{n^2-1} ⋯ a_0$ be its  binary expansion
  (the least significant digit is $a_0$).
  For each $i,j$ such that $i,j=0,\ldots, n-1$ the $(i,j)$-entry of the matrix~$N_k$ is $a_{n^2 - 1 - in - j}$.
The $(0,0)$-entry of $N_k$ is thus
 the first digit $a_{n^2-1}$ and the $(n-1,n-1)$-entry is the last digit~$a_0$.
\end{definition}

For example  the enumeration  
$N_0,…,N_{15}$ of the $16$ matrices 
over~$\FT$ of size $2 × 2$ is
\medskip

\begin{align*}
  & \left(
    \begin{smallmatrix}
      0 & 0 \\
      0 & 0 
    \end{smallmatrix}
  \right),
  \left(
    \begin{smallmatrix}
      0 & 0 \\
      0 & 1 
    \end{smallmatrix}
  \right),
  \left(
    \begin{smallmatrix}
      0 & 0 \\
      1 & 0 
    \end{smallmatrix}
  \right),
  \left(
    \begin{smallmatrix}
      0 & 0 \\
      1 & 1 
    \end{smallmatrix}
  \right),
  \left(
    \begin{smallmatrix}
      0 & 1 \\
      0 & 0 
    \end{smallmatrix}
  \right),
  \left(
    \begin{smallmatrix}
      0 & 1 \\
      0 & 1 
    \end{smallmatrix}
  \right),
  \left(
    \begin{smallmatrix}
      0 & 1 \\
      1 & 0 
    \end{smallmatrix}
  \right),
  \left(
    \begin{smallmatrix}
      0 & 1 \\
      1 & 1 
    \end{smallmatrix}
  \right), \\
  & \left(
    \begin{smallmatrix}
      1 & 0 \\
      0 & 0 
    \end{smallmatrix}
  \right),
  \left(
    \begin{smallmatrix}
      1 & 0 \\
      0 & 1 
    \end{smallmatrix}
  \right),
  \left(
    \begin{smallmatrix}
      1 & 0 \\
      1 & 0 
    \end{smallmatrix}
  \right),
  \left(
    \begin{smallmatrix}
      1 & 0 \\
      1 & 1 
    \end{smallmatrix}
  \right),
  \left(
    \begin{smallmatrix}
      1 & 1 \\
      0 & 0 
    \end{smallmatrix}
  \right),
  \left(
    \begin{smallmatrix}
      1 & 1 \\
      0 & 1 
    \end{smallmatrix}
  \right),
  \left(
    \begin{smallmatrix}
      1 & 1 \\
      1 & 0 
    \end{smallmatrix}
  \right),
  \left(
    \begin{smallmatrix}
      1 & 1 \\
      1 & 1 
    \end{smallmatrix}
  \right)
\end{align*}
\medskip

For any given non negative integer $k$ we consider its binary representation  as a sequence of bits $a_n \ldots a_0$. We denote with $\even(k)$  the integer whose binary representation is the subsequence 
$a_m\ldots a_2a_0$,  made of the bits   at even positions in the representation of~$k$, so  where $m=n$ in case $n$ is even, otherwise $m=n-1$. Similarly,  $\odd(k)$ is the integer  defined from the subsequence determined  by  the odd indexes.
\medskip

\begin{definition}[Pascal toroidal array]\label{def:Pascal}
Let  $d$ be a positive integer and let $n=2^d$.
 We define a toroidal array~$A_d$ of size $n2^{n^2/2} × n2^{n^2/2}$ over~$\FT$   by
tiling it with all the   $n × n$ matrices over~$\FT$.  Since there are
$2^{n^2}$ such matrices, the total number of placed cells is exactly the
size of~$A_d$.   
For each $k$ such that $0 ⩽ k <2^{n^2}$
 the
matrix $M_dN_k$ is placed in $A_d$
at position $(\odd(k)n, \even(k)n)$. 
\end{definition}

Since each  matrix~$M_d$   is upper
triangular with $1$s on the diagonal then 
it is invertible. 
Since $N_0,…, N_{n^2-1}$ is an enumeration of all $n × n$ matrices, then
$M_dN_0,…,M_dN_{n^2-1}$ is also an enumeration of all the $n × n$
matrices, but in a different order. 
This implies that each $n × n$
matrix is  used  exactly once to tile the array~$A_d$.

We illustrate the construction of $A_d$ for $d = 1$, $n=2$. 
The $2 × 2$ matrices $N_0, \ldots , N_{15}$ are listed above. 
The matrices $M_1 N_k$ for $k=0, \ldots , 15$ are placed as follows in the array $A_1$:

\begin{displaymath}
  \begin{array}{cccc}
    M_1N_0    & M_1N_1   & M_1N_4   & M_1N_5 \\
    M_1N_2    & M_1N_3   & M_1N_6   & M_1N_7 \\
    M_1N_8    & M_1N_9   & M_1N_{12} & M_1N_{13} \\
    M_1N_{10} & M_1N_{11} & M_1N_{14} & M_1N_{15} 
  \end{array}
\end{displaymath}
\medskip

Since each matrix $M_1N_k$ has size $2 × 2$, the array~$A_1$ has size $8 × 8$.  
The nested $(2,2,2,2)$-perfect array given in the right of
Figure~\ref{fig:perfect} is actually the array~$A_1$.

\begin{figure}[tb]
  \begin{center}
    \begin{tikzpicture}
      \draw (-0.5,4) -- (4.5,4);
      \draw (-0.5,2) -- (4.5,2);
      \draw (-0.5,0) -- (4.5,0);
      \draw (0,-0.5) -- (0,4.5);
      \draw (2,-0.5) -- (2,4.5);
      \draw (4,-0.5) -- (4,4.5);
      \draw (0.8,2.5) -- (2.8,2.5) -- (2.8,1) -- (0.8,1) -- cycle;
      \draw[dashed] (-0.3,2.5) -- (0.8,2.5);
      \draw[<->] (-0.2,2.5) -- (-0.2,4);
      \node at (-0.4,3.2) {$i$};
      \draw[dashed] (-0.3,0.5) -- (0.8,0.5);
      \draw[<->] (-0.2,0.5) -- (-0.2,2);
      \node at (-0.4,1.2) {$i$};
      \draw[dashed] (0.8,4.3)  -- (0.8,2.5);
      \draw[<->] (0,4.2) -- (0.8,4.2);
      \node at (0.4,4.4) {$j$};
      \draw[dashed] (2.8,4.3)  -- (2.8,2.5);
      \draw[<->] (2,4.2) -- (2.8,4.2);
      \node at (2.4,4.4) {$j$};
      \draw[dashed] (0.4,1) -- (0.8,1);
      \draw[<->] (0.6,1) -- (0.6,2.5);
      \node at (0.4,1.7) {$k$};
      \draw[<->] (0.6,0.5) -- (0.6,1);
      \node at (1.1,0.75) {$n{-}k$};
      \draw[dashed] (2.8,2.5) -- (3.2,2.5);
      \draw[<->] (3,2) -- (3,2.5);
      \node at (3.45,2.25) {$n{-}i$};
      \draw[<->] (0.8,3.3) -- (2.8,3.3);
      \node at (1.8,3.5) {$n$};
      
      \draw[<->] (0,-0.2) -- (2,-0.2);
      \node at (1,-0.5) {$n$};
      \draw[<->] (2,-0.2) -- (4,-0.2);
      \node at (3,-0.5) {$n$};
      \draw[<->] (4.2,0) -- (4.2,2);
      \node at (4.5,1) {$n$};
      \draw[<->] (4.2,2) -- (4.2,4);
      \node at (4.5,3) {$n$};
      \node at (1.4,2.2) {$P_1$};
      \node at (2.4,2.2) {$P_2$};
      \node at (1.4,1.5) {$P_3$};
      \node at (2.4,1.5) {$P_4$};
    \end{tikzpicture}
  \end{center}
  \caption{An occurrence of array $B$ in $A_d$,  $n=2^d$.}
  \label{fig:occurrence}
\end{figure}
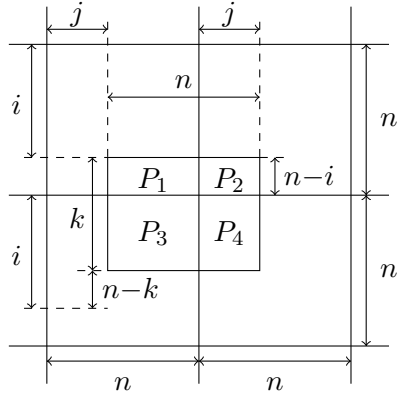

\begin{proposition}\label{prop:main}
Let  $d$ be a non-negative integer and let $n=2^d$.
  Let $k$ be an integer such that  $1 ⩽ k ⩽ n$. 
  Each $n2^{kn/2} × n2^{kn/2}$
  aligned subarray of~$A_d$ is a nested $(k,n,n,n)$-perfect toroidal
  array.
\end{proposition}

  \begin{figure}[t!]
    \begin{subfigure}{0.8\textwidth}
      \begin{tikzpicture}[scale=1.2]
        \begin{scope}
          \draw[dashed] (0,0) -- (0,2) -- (2,2) -- (2,0) -- cycle;
          \draw (0.8,0.5) -- (2,0.5) -- (2,0) -- (0.8,0) -- cycle;
          \draw[dashed] (-0.3,2) -- (0,2);
          \draw[dashed] (-0.3,0.5) -- (0.8,0.5);
          \draw[<->] (-0.2,0.5) -- (-0.2,2);
          \node at (-0.4,1.2) {$i$};
          \draw[dashed] (0,2)  -- (0,2.3);
          \draw[dashed] (0.8,2.3)  -- (0.8,0.5);
          \draw[<->] (0,2.2) -- (0.8,2.2);
          \node at (0.4,2.4) {$j$};
          \node at (1.4,0.2) {$P_1$};
        \end{scope}
        \node at (3,1) {$=$};
        \begin{scope}[shift={(4,0)}]
          \draw (0,0) -- (0,2) -- (2,2) -- (2,0) -- cycle;
          \draw[fill=gray!20] (0,0.5) -- (2,0.5) -- (2,0) -- (0,0) -- cycle;
          \draw[dashed] (-0.3,2) -- (0,2);
          \draw[dashed] (-0.3,0.5) -- (0,0.5);
          \draw[<->] (-0.2,0.5) -- (-0.2,2);
          \node at (-0.4,1.2) {$i$};
          \node at (1,0.95) {$M_d$};
        \end{scope}
        \node at (6.5,1) {$×$};
        \begin{scope}[shift={(7,0)}]
          \draw (0,0) -- (0,2) -- (2,2) -- (2,0) -- cycle;
          \draw[fill=gray!20] (0.8,0) -- (0.8,2) -- (2,2) -- (2,0) -- cycle;
          \draw[dashed] (0,2)  -- (0,2.3);
          \draw[dashed] (0.8,2.3)  -- (0.8,0.5);
          \draw[<->] (0,2.2) -- (0.8,2.2);
          \node at (0.4,2.4) {$j$};
          \node at (1.4,0.95) {$N_{\join{ℓ}{n}}$};
        \end{scope}
      \end{tikzpicture}
      \caption{Parts of $M_d$ and $N_{\join{ℓ}{n}}$ used to compute $P_1$}
      \label{fig:p1part1}
    \end{subfigure} 
    \begin{subfigure}{0.8\textwidth}
      \begin{tikzpicture}[scale=1.2]
        \begin{scope}
          \draw[dashed] (0,0) -- (0,2) -- (2,2) -- (2,0) -- cycle;
          \draw (0.8,0.5) -- (2,0.5) -- (2,0) -- (0.8,0) -- cycle;
          \draw[dashed] (-0.3,2) -- (0,2);
          \draw[dashed] (-0.3,0.5) -- (0.8,0.5);
          \draw[<->] (-0.2,0.5) -- (-0.2,2);
          \node at (-0.4,1.2) {$i$};
          \draw[dashed] (0,2)  -- (0,2.3);
          \draw[dashed] (0.8,2.3)  -- (0.8,0.5);
          \draw[<->] (0,2.2) -- (0.8,2.2);
          \node at (0.4,2.4) {$j$};
          \node at (1.4,0.2) {$P_1$};
        \end{scope}
        \node at (3,1) {$=$};
        \begin{scope}[shift={(4,0)}]
          \draw (0,0) -- (0,2) -- (2,2) -- (2,0) -- cycle;
          \draw (2,0) -- (0,2);
          \draw[fill=gray!20] (2,0) -- (1.5,0.5) -- (2,0.5) -- cycle;
          \draw[dashed] (-0.3,2) -- (0,2);
          \draw[dashed] (-0.3,0.5) -- (1.5,0.5);
          \draw[<->] (-0.2,0.5) -- (-0.2,2);
          \node at (-0.4,1.2) {$i$};
          \node at (1.3,1.3) {$M_d$};
        \end{scope}
        \node at (6.5,1) {$×$};
        \begin{scope}[shift={(7,0)}]
          \draw (0,0) -- (0,2) -- (2,2) -- (2,0) -- cycle;
          \draw[fill=gray!20] (0.8,0) -- (0.8,0.5) -- (2,0.5) -- (2,0) -- cycle;
          \draw[fill=gray] (0,1.5) -- (0,2) -- (2,2) -- (2,1.5) -- cycle;
          \node[white] at (1,1.75) {FIXED};
          \draw[dashed] (2,2) -- (2.3,2);
          \draw[dashed] (2,1.5) -- (2.3,1.5);
          \draw[<->] (2.2,2) -- (2.2,1.5);
          \node at (2.65,1.75) {$n{-}k$};
          \node at (1,0.95) {$N_{\join{ℓ}{n}}$};
        \end{scope}
      \end{tikzpicture}
      \caption{Taking into account that $M_d$ is upper triangular}
      \label{fig:p1part2}
    \end{subfigure}
    \caption{Proof of Proposition~\ref{prop:main},  using $P_1$}
    \label{fig:p1}
    \begin{subfigure}{0.8\textwidth}
      \begin{tikzpicture}[scale=1.2]
        \begin{scope}
          \draw[dashed] (0,0) -- (0,2) -- (2,2) -- (2,0) -- cycle;
          \draw (0,0) -- (0,0.5) -- (0.8,0.5) -- (0.8,0) -- cycle;
          \draw[dashed] (-0.3,2) -- (0,2);
          \draw[dashed] (-0.3,0.5) -- (0.8,0.5);
          \draw[<->] (-0.2,0.5) -- (-0.2,2);
          \node at (-0.4,1.2) {$i$};
          \draw[dashed] (0,2)  -- (0,2.3);
          \draw[dashed] (0.8,2.3)  -- (0.8,0.5);
          \draw[<->] (0,2.2) -- (0.8,2.2);
          \node at (0.4,2.4) {$j$};
          \node at (0.4,0.2) {$P_2$};
        \end{scope}
        \node at (3,1) {$=$};
        \begin{scope}[shift={(4,0)}]
          \draw (0,0) -- (0,2) -- (2,2) -- (2,0) -- cycle;
          \draw[fill=gray!20] (0,0.5) -- (2,0.5) -- (2,0) -- (0,0) -- cycle;
          \draw[dashed] (-0.3,2) -- (0,2);
          \draw[dashed] (-0.3,0.5) -- (0,0.5);
          \draw[<->] (-0.2,0.5) -- (-0.2,2);
          \node at (-0.4,1.2) {$i$};
          \node at (1,0.95) {$M_d$};
        \end{scope}
        \node at (6.5,1) {$×$};
        \begin{scope}[shift={(7,0)}]
          \draw (0,0) -- (0,2) -- (2,2) -- (2,0) -- cycle;
          \draw[fill=gray!20] (0,0) -- (0,2) -- (0.8,2) -- (0.8,0) -- cycle;
          \draw[dashed] (0,2)  -- (0,2.3);
          \draw[dashed] (0.8,2.3)  -- (0.8,0.5);
          \draw[<->] (0,2.2) -- (0.8,2.2);
          \node at (0.4,2.4) {$j$};
          \node at (1,0.95) {$N_{\join{ℓ}{(n{+}1)}}$};
        \end{scope}
      \end{tikzpicture}
      \caption{Parts of $M_d$ and $N_{\join{ℓ}{n}}$ used to compute $P_2$}
      \label{fig:p2part1}
    \end{subfigure} 
    \begin{subfigure}{0.8\textwidth}
      \begin{tikzpicture}[scale=1.2]
        \begin{scope}
          \draw[dashed] (0,0) -- (0,2) -- (2,2) -- (2,0) -- cycle;
          \draw (0,0) -- (0,0.5) -- (0.8,0.5) -- (0.8,0) -- cycle;
          \draw[dashed] (-0.3,2) -- (0,2);
          \draw[dashed] (-0.3,0.5) -- (0.8,0.5);
          \draw[<->] (-0.2,0.5) -- (-0.2,2);
          \node at (-0.4,1.2) {$i$};
          \draw[dashed] (0,2)  -- (0,2.3);
          \draw[dashed] (0.8,2.3)  -- (0.8,0.5);
          \draw[<->] (0,2.2) -- (0.8,2.2);
          \node at (0.4,2.4) {$j$};
          \node at (0.4,0.2) {$P_2$};
        \end{scope}
        \node at (3,1) {$=$};
        \begin{scope}[shift={(4,0)}]
          \draw (0,0) -- (0,2) -- (2,2) -- (2,0) -- cycle;
          \draw (2,0) -- (0,2);
          \draw[fill=gray!20] (2,0) -- (1.5,0.5) -- (2,0.5) -- cycle;
          \draw[dashed] (-0.3,2) -- (0,2);
          \draw[dashed] (-0.3,0.5) -- (1.5,0.5);
          \draw[<->] (-0.2,0.5) -- (-0.2,2);
          \node at (-0.4,1.2) {$i$};
          \node at (1.3,1.3) {$M_d$};
        \end{scope}
        \node at (6.5,1) {$×$};
        \begin{scope}[shift={(7,0)}]
          \draw (0,0) -- (0,2) -- (2,2) -- (2,0) -- cycle;
          \draw[fill=gray!20] (0,0) -- (0,0.5) -- (0.8,0.5) -- (0.8,0) -- cycle;
          \draw[fill=gray] (0,1.5) -- (0,2) -- (2,2) -- (2,1.5) -- cycle;
          \node[white] at (1,1.75) {FIXED};
          \draw[dashed] (2,2) -- (2.3,2);
          \draw[dashed] (2,1.5) -- (2.3,1.5);
          \draw[<->] (2.2,2) -- (2.2,1.5);
          \node at (2.65,1.75) {$n{-}k$};
          \node at (1,0.95) {$N_{\join{ℓ}{(n{+}1)}}$};
        \end{scope}
      \end{tikzpicture}
      \caption{Taking into account that $M_d$ is upper triangular}
      \label{fig:p2part2}
    \end{subfigure}
    \caption{Proof of Proposition~\ref{prop:main},  using $P_2$}
    \label{fig:p2}
  \end{figure}

  \begin{figure}[t!]
    \begin{subfigure}{0.95\textwidth}
      \begin{tikzpicture}[scale=1.2]
        \begin{scope}
          \draw[dashed] (0,0) -- (0,2) -- (2,2) -- (2,0) -- cycle;
          \draw (0.8,1) -- (0.8,2) -- (2,2) -- (2,1) -- cycle;
          \draw[dashed] (-0.3,2) -- (0,2);
          \draw[dashed] (-0.3,1)  -- (0.8,1);
          \draw[<->] (-0.2,1) -- (-0.2,2);
          \node at (-0.8,1.5) {$i{+}k{-}n$};
          \draw[dashed] (0,2)  -- (0,2.3);
          \draw[dashed] (0.8,2)  -- (0.8,2.3);
          \draw[<->] (0,2.2) -- (0.8,2.2);
          \node at (0.4,2.4) {$j$};
          \node at (1.4,1.5) {$P_3$};
        \end{scope}
        \node at (3,1) {$=$};
        \begin{scope}[shift={(4,0)}]
          \draw (0,0) -- (0,2) -- (2,2) -- (2,0) -- cycle;
          \draw[fill=gray!20] (0,1) -- (0,2) -- (2,2) -- (2,1) -- cycle;
          \draw[dashed] (-0.3,2) -- (0,2);
          \draw[dashed] (-0.3,1)  -- (0,1);
          \draw[<->] (-0.2,1) -- (-0.2,2);
          \node at (-0.8,1.5) {$i{+}k{-}n$};
          \node at (1,0.65) {$M_d$};
        \end{scope}
        \node at (6.5,1) {$×$};
        \begin{scope}[shift={(7,0)}]
          \draw (0,0) -- (0,2) -- (2,2) -- (2,0) -- cycle;
          \draw[fill=gray!20] (0.8,0) -- (0.8,2) -- (2,2) -- (2,0) -- cycle;
          \draw[dashed] (0,2)  -- (0,2.3);
          \draw[dashed] (0.8,2.3)  -- (0.8,0.5);
          \draw[<->] (0,2.2) -- (0.8,2.2);
          \node at (0.4,2.4) {$j$};
          \node at (1,0.95) {$N_{\join{(ℓ{+}1)}{n}}$};
        \end{scope}
      \end{tikzpicture}
      \caption{Parts of $M_d$ and $N_{\join{ℓ}{n}}$ used to compute $P_3$}
      \label{fig:p3part1}
    \end{subfigure} 
    \begin{subfigure}{0.95\textwidth}
      \begin{tikzpicture}[scale=1.2]
        \begin{scope}
          \draw[dashed] (0,0) -- (0,2) -- (2,2) -- (2,0) -- cycle;
          \draw (0.8,1) -- (0.8,2) -- (2,2) -- (2,1) -- cycle;
          \draw[dashed] (-0.3,2) -- (0,2);
          \draw[dashed] (-0.3,1)  -- (0.8,1);
          \draw[<->] (-0.2,1) -- (-0.2,2);
          \node at (-0.8,1.5) {$i{+}k{-}n$};
          \draw[dashed] (0,2)  -- (0,2.3);
          \draw[dashed] (0.8,2)  -- (0.8,2.3);
          \draw[<->] (0,2.2) -- (0.8,2.2);
          \node at (0.4,2.4) {$j$};
          \node at (1.4,1.5) {$P_3$};
        \end{scope}
        \node at (3,1) {$=$};
        \begin{scope}[shift={(4,0)}]
          \draw (0,0) -- (0,2) -- (2,2) -- (2,0) -- cycle;
          \draw (2,0) -- (0,2);
          \draw[fill=gray] (0,2) -- (0.5,2) -- (0.5,1.5) -- cycle;
          \draw[fill=gray!20] (0.5,2) -- (1.5,2) -- (1.5,1) -- (1,1) --
                              (0.5,1.5) -- cycle;
          \draw[fill=gray] (1.5,1) -- (1.5,2) -- (2,2) -- (2,1) -- cycle;
          \draw[dashed] (0,2)  -- (0,2.3);
          \draw[dashed] (0.5,2)  -- (0.5,2.3);
          \draw[<->] (0,2.2) -- (0.5,2.2);
          \node at (0.25,2.4) {$n{-}k$};
          \draw[dashed] (1.5,2)  -- (1.5,2.3);
          \draw[dashed] (2,2)  -- (2,2.3);
          \draw[<->] (1.5,2.2) -- (2,2.2);
          \node at (1.75,2.4) {$i$};
          \draw[dashed] (-0.3,2) -- (0,2);
          \draw[dashed] (-0.3,1)  -- (0.8,1);
          \draw[<->] (-0.2,1) -- (-0.2,2);
          \node at (-0.8,1.5) {$i{+}k{-}n$};
          \node at (1,1.5) {$M_d$};
        \end{scope}
        \node at (6.5,1) {$×$};
        \begin{scope}[shift={(7,0)}]
          \draw (0,0) -- (0,2) -- (2,2) -- (2,0) -- cycle;
          \draw[fill=gray!20] (0.8,0.5) -- (0.8,1.5) -- (2,1.5) -- (2,0.5) -- cycle;
          \draw[fill=gray] (0,1.5) -- (0,2) -- (2,2) -- (2,1.5) -- cycle;
          \node[white] at (1,1.75) {FIXED};
          \draw[fill=gray] (0,0) -- (0,0.5) -- (2,0.5) -- (2,0) -- cycle;
          \node[white] at (1,0.25) {KNOWN};
          \draw[dashed] (2,0.5) -- (2.3,0.5);
          \draw[dashed] (2,0) -- (2.3,0);
          \draw[<->] (2.2,0) -- (2.2,0.5);
          \node at (2.4,0.25) {$i$};
          \draw[dashed] (2,2) -- (2.3,2);
          \draw[dashed] (2,1.5) -- (2.3,1.5);
          \draw[<->] (2.2,2) -- (2.2,1.5);
          \node at (2.65,1.75) {$n{-}k$};
          \node at (1,0.95) {$N_{\join{(ℓ{+}1)}{n}}$};
        \end{scope}
      \end{tikzpicture}
      \caption{Taking into account that $M_d$ is upper triangular}
      \label{fig:p3part2}
    \end{subfigure}
    \caption{Proof of Proposition~\ref{prop:main},  using $P_3$}
    \label{fig:p3}
   \begin{subfigure}{0.95\textwidth}
      \begin{tikzpicture}[scale=1.2]
        \begin{scope}
          \draw[dashed] (0,0) -- (0,2) -- (2,2) -- (2,0) -- cycle;
          \draw (0,1) -- (0,2) -- (0.8,2) -- (0.8,1) -- cycle;
          \draw[dashed] (-0.3,2) -- (0,2);
          \draw[dashed] (-0.3,1)  -- (0,1);
          \draw[<->] (-0.2,1) -- (-0.2,2);
          \node at (-0.8,1.5) {$i{+}k{-}n$};
          \draw[dashed] (0,2)  -- (0,2.3);
          \draw[dashed] (0.8,2)  -- (0.8,2.3);
          \draw[<->] (0,2.2) -- (0.8,2.2);
          \node at (0.4,2.4) {$j$};
          \node at (0.4,1.5) {$P_4$};
        \end{scope}
        \node at (3,1) {$=$};
        \begin{scope}[shift={(4,0)}]
          \draw (0,0) -- (0,2) -- (2,2) -- (2,0) -- cycle;
          \draw[fill=gray!20] (0,1) -- (0,2) -- (2,2) -- (2,1) -- cycle;
          \draw[dashed] (-0.3,2) -- (0,2);
          \draw[dashed] (-0.3,1)  -- (0,1);
          \draw[<->] (-0.2,1) -- (-0.2,2);
          \node at (-0.8,1.5) {$i{+}k{-}n$};
          \node at (1,0.65) {$M_d$};
        \end{scope}
        \node at (6.5,1) {$×$};
        \begin{scope}[shift={(7,0)}]
          \draw (0,0) -- (0,2) -- (2,2) -- (2,0) -- cycle;
          \draw[fill=gray!20] (0,0) -- (0,2) -- (0.8,2) -- (0.8,0) -- cycle;
          \draw[dashed] (0,2)  -- (0,2.3);
          \draw[dashed] (0.8,2.3)  -- (0.8,0.5);
          \draw[<->] (0,2.2) -- (0.8,2.2);
          \node at (0.4,2.4) {$j$};
          \node at (1,0.95) {$N_{\join{(ℓ{+}1)}{n}}$};
        \end{scope}
      \end{tikzpicture}
      \caption{Parts of $M_d$ and $N_{\join{ℓ}{n}}$ used to compute $P_4$}
      \label{fig:p4part1}
    \end{subfigure} 
    \begin{subfigure}{0.95\textwidth}
      \begin{tikzpicture}[scale=1.2]
        \begin{scope}
          \draw[dashed] (0,0) -- (0,2) -- (2,2) -- (2,0) -- cycle;
          \draw (0,1) -- (0,2) -- (0.8,2) -- (0.8,1) -- cycle;
          \draw[dashed] (-0.3,2) -- (0,2);
          \draw[dashed] (-0.3,1)  -- (0,1);
          \draw[<->] (-0.2,1) -- (-0.2,2);
          \node at (-0.8,1.5) {$i{+}k{-}n$};
          \draw[dashed] (0,2)  -- (0,2.3);
          \draw[dashed] (0.8,2)  -- (0.8,2.3);
          \draw[<->] (0,2.2) -- (0.8,2.2);
          \node at (0.4,2.4) {$j$};
          \node at (0.4,1.5) {$P_4$};
        \end{scope}
        \node at (3,1) {$=$};
        \begin{scope}[shift={(4,0)}]
          \draw (0,0) -- (0,2) -- (2,2) -- (2,0) -- cycle;
          \draw (2,0) -- (0,2);
          \draw[fill=gray] (0,2) -- (0.5,2) -- (0.5,1.5) -- cycle;
          \draw[fill=gray!20] (0.5,2) -- (1.5,2) -- (1.5,1) -- (1,1) --
                              (0.5,1.5) -- cycle;
          \draw[fill=gray] (1.5,1) -- (1.5,2) -- (2,2) -- (2,1) -- cycle;
          \draw[dashed] (0,2)  -- (0,2.3);
          \draw[dashed] (0.5,2)  -- (0.5,2.3);
          \draw[<->] (0,2.2) -- (0.5,2.2);
          \node at (0.25,2.4) {$n{-}k$};
          \draw[dashed] (1.5,2)  -- (1.5,2.3);
          \draw[dashed] (2,2)  -- (2,2.3);
          \draw[<->] (1.5,2.2) -- (2,2.2);
          \node at (1.75,2.4) {$i$};
          \draw[dashed] (-0.3,2) -- (0,2);
          \draw[dashed] (-0.3,1)  -- (0.8,1);
          \draw[<->] (-0.2,1) -- (-0.2,2);
          \node at (-0.8,1.5) {$i{+}k{-}n$};
          \node at (1,1.5) {$M_d$};
        \end{scope}
        \node at (6.5,1) {$×$};
        \begin{scope}[shift={(7,0)}]
          \draw (0,0) -- (0,2) -- (2,2) -- (2,0) -- cycle;
          \draw[fill=gray!20] (0,0.5) -- (0,1.5) -- (0.8,1.5) -- (0.8,0.5) -- cycle;
          \draw[fill=gray] (0,1.5) -- (0,2) -- (2,2) -- (2,1.5) -- cycle;
          \node[white] at (1,1.75) {FIXED};
          \draw[fill=gray] (0,0) -- (0,0.5) -- (2,0.5) -- (2,0) -- cycle;
          \node[white] at (1,0.25) {KNOWN};
          \draw[dashed] (2,0.5) -- (2.3,0.5);
          \draw[dashed] (2,0) -- (2.3,0);
          \draw[<->] (2.2,0) -- (2.2,0.5);
          \node at (2.4,0.25) {$i$};
          \draw[dashed] (2,2) -- (2.3,2);
          \draw[dashed] (2,1.5) -- (2.3,1.5);
          \draw[<->] (2.2,2) -- (2.2,1.5);
          \node at (2.65,1.75) {$n{-}k$};
          \node at (1,0.95) {$N_{\join{(ℓ{+}1)}{(n{+}1)}}$};
        \end{scope}
      \end{tikzpicture}
      \caption{Taking into account that $M_d$ is upper triangular}
      \label{fig:p4part2}
    \end{subfigure}
    \caption{Proof of Proposition~\ref{prop:main}, using $P_4$}
    \label{fig:p4}
  \end{figure}
   
\begin{proof}
  Suppose that $k$ is fixed,  $1 ⩽ k ⩽ n$, and let $B$ be an aligned subarray
  of~$A_d$ of sizes $n2^{kn/2} × n2^{kn/2}$.  Since $B$ is aligned, the
  coordinates of its upper left corner are of the form $pn2^{kn/2}$ and
  $qn2^{kn/2}$ for two integers $p$ and~$q$ such that $0 ⩽ p,q <
  2^{(n-k)n/2}$.  This means that the subarray~$B$ is tiled by the
  matrices $M_d N_{\join{ℓ}{m}}$ for integers $ℓ$ and~$m$ satisfying
  \[p2^{kn/2} ⩽ ℓ < (p+1)2^{kn/2} \text{ and }
 q2^{kn/2}⩽ m < (q+1)2^{kn/2}.\]  
  Note
  that the factor~$n$ disappeared since it accounts for the sizes of each
  of the matrices $M_d N_{\join{ℓ}{m}}$.  The binary expansions of all
  integers~$ℓ$ satisfying 
$p2^{kn/2} ⩽ ℓ < (p+1)2^{kn/2}$
start with the
  same $(n-k)n/2$ binary digits and the same hold for all integers~$m$
  satisfying 
  $q2^{kn/2} ⩽ m < (q+1)2^{kn/2}.$  
  This implies that the binary
  expansion of $\join{ℓ}{m}$ starts with the same $(n-k)n$ binary digits.
  Since the first binary digits of $\join{ℓ}{m}$ are put in the first rows
  of the matrix~$N_{\join{ℓ}{m}}$ which have length~$n$, all matrices
  $N_{\join{ℓ}{m}}$ for $ℓ$ and $m$ in their respective intervals have the
  same first $n - k$ rows.
  
  Let $(i,j)$ be a pair of integers such that $0 ⩽ i,j < n$ and let $P$ be
  an array of size $k × n$.  We claim that $P$ has exactly one occurrence
  in~$B$ which is congruent to
  $(i, j)$ modulo~$(n, n)$.
  To prove it, we show that $P$ has a single such occurrence exactly when a certain
  system of linear equations has a solution. 
 Furthermore, this solution of
  the system provides the matrix~$N_{\join{ℓ}{m}}$ and thus the integers $ℓ$
  and~$m$ which, in turn, give the position of the occurrence of~$P$ in the  subarray~$B$.  
\medskip

An occurrence~$P$ can overlap at most four matrices tiling
  the subarray~$B$.  Suppose that the upper left corner of the occurrence
  of~$P$ lies in some matrix $M_d N_{\join{ℓ}{m}}$ where the integers $ℓ$
  and~$m$ such that $p2^{kn/2} ⩽ ℓ < (p+1)2^{kn/2}$ and $q2^{kn/2} ⩽ m <
  (q+1)2^{kn/2}$.  The matrix on the right of~$M_d N_{\join{ℓ}{m}}$ and the
  matrix below it are respectively $M_d N_{\join{(ℓ+1)}{m}}$ and $M_d
  N_{\join{ℓ}{(m+1)}}$ where $ℓ+1$ and $m+1$ must be understood modulo
  $2^{kn/2}$ in order to remain in the right intervals.  Let $P_1$, $P_2$,
  $P_3$ and~$P_4$ be the parts of~$P$ that overlap respectively the
  matrices $M_d N_{\join{ℓ}{m}}$, $M_d N_{\join{ℓ}{(m+1)}}$, $M_d
  N_{\join{(ℓ+1)}{m}}$ and $M_d N_{\join{(ℓ+1)}{(m+1)}}$.  They are
  pictured in Figure~\ref{fig:occurrence}.
  
  If $j = 0$, the parts $P_2$ and~$P_4$ of the occurrence are empty.  If $i
  + k ⩽ n$, the parts $P_3$ and~$P_4$ of the occurrence are also empty.
  The simplest case is  the system  of equations 
  $M_d N_{\join{\ell}{m}} = P$ when $i = j = 0$ and $k = n$.
  We now treat the general case where none are empty, the other cases are similar
  and easier.

  We state now the system of equations.
  The unknowns are the entries in the matrix $N_{\join{ℓ}{m}}$. 
  We claim that  they can be found from the matrix $A_d$
  the pair $(i,j)$ and the arrays~$P_1$,$P_2$,$P_3$ and $P_4$.  
  As explained before, since  $k=\join{ℓ}{m}$,
  the first $n-k$ rows of the matrix $N_{\join{ℓ}{m}}$ are fixed by the
  subarray~$B$. 
  This part is marked in dark grey in Figure~\ref{fig:p1}.
  
  The height and width of~$P_1$ are respectively $n-i$ and $n-j$.  The
  part~$P_1$ is obtained by the multiplication of the last $n-i$ rows of
  the matrix~$M_d$ and the last $n-j$ columns of the
  matrix~$N_{\join{ℓ}{m}}$ (see grey parts in Figure~\ref{fig:p1part1}).
  Now we use the fact that the matrix~$M_d$ is upper triangular.  This
  reduces the parts of $M_d$ and~$N_{\join{ℓ}{m}}$ used to compute~$P_1$
  (see grey parts in Figure~\ref{fig:p1part2}).  Furthermore, the part used
  in~$M_d$ is upper triangular matrix with~$1$ on the diagonal.  This
  matrix is therefore invertible.  This means that the grey part
  in~$N_{\join{ℓ}{m}}$ can be obtained by multiplying the inverse of this
  triangular matrix with~$P_1$.

  The height and width of~$P_2$ are respectively $n-i$ and~$j$.  The
  part~$P_2$ is obtained by the multiplication of the last $n-i$ rows of
  the matrix~$M_d$ and the first $j$ columns of the
  matrix~$N_{\join{ℓ}{(m+1)}}$ (see grey parts in
  Figure~\ref{fig:p2part1}).  As for~$P_1$, the fact that the matrix~$M_d$
  is upper triangular reduces the parts of $M_d$ and~$N_{\join{ℓ}{(m+1)}}$
  used to compute~$P_2$ (see grey parts in Figure~\ref{fig:p2part2}).
  Furthermore, the part used in~$M_d$ is again invertible.  This means that
  the grey part in~$N_{\join{ℓ}{(m+1)}}$ can be obtained by multiplying the
  inverse of this triangular matrix with~$P_2$.
  
  We claim that $n-i$ rows of~$N_{\join{ℓ}{m}}$ are determined by part
  known in $N_{\join{ℓ}{m}}$ and~$N_{\join{ℓ}{(m+1)}}$.  This is because,
  for integers $r$ and~$s$, the last $r$ binary digits of~$s$ determine the
  last $r$ binary digits of~$s+1$ and that conversely the last $r$ binary
  digits if~$s+1$ determine the last $r$ binary digits of~$s$.  It follows
  that the last $i$ rows of the four matrices $N_{\join{ℓ}{m}}$,
  $N_{\join{ℓ}{(m+1)}}$, $N_{\join{(ℓ+1)}{m}}$ and
  $N_{\join{(ℓ+1)}{(m+1)}}$ are known.  These parts are marked in dark grey
  in the Figure~\ref{fig:p2}.

  The height and width of~$P_3$ are respectively $i+k-n$ and $n-j$.
  The part~$P_3$ is obtained by the multiplication of the first $i+k-n$
  rows of the matrix~$M_d$ and the last $n-j$ columns of the
  matrix~$N_{\join{(ℓ+1)}{m}}$ (see grey parts in
  Figure~\ref{fig:p3part1}).  Now we use the fact that the first $n-k$
  and the last $i$ rows of $N_{\join{(ℓ+1)}{m}}$ are known.  The elements
  of these rows can be considered as constants in the system of equations.
  Combining this latter result and the fact that the square $(i+k-n) ×
  (i+k-n)$ submatrix of~$M_d$ in  grey in Figure~\ref{fig:p3part2})
  is invertible 
  the still unknown entries in the last $n-j$
  columns of~$N_{\join{(ℓ+1)}{m}}$ can be found.
  
  The height and width of~$P_4$ are respectively $i+k-n$ and~$j$.  By a
  reasoning very similar used with~$P_3$, the remaining entries of the
  matrix~$N_{\join{(ℓ+1)}{(m+1)}}$ can be found, see
  Figures \ref{fig:p4part1} and~\ref{fig:p4part2}.
\end{proof}

In Proposition~\ref{prop:main} the case $k = n$ states that the  Pascal toroidal array for $n=2^d$  is a nested 
$(n,n,n,n)$-perfect toroidal array.
This  is proves Theorem~\ref{thm:existence}.

\section{Proof of Theorem~\ref{thm:constructions}}

We consider a family of  matrices   that were first defined in~\cite{BecherCarton2019}.  
These matrices are obtained by applying some rotations
to columns of the matrices~$M_d$  given in 
Definition~\ref{def:Md}.

Let $\sigma$ be the function which maps
each word $a_1 \cdots a_n$ to $a_na_1a_2 \cdots a_{n-1}$ obtained by moving
the last symbol to the front.  Since words over~$\FT$ are identified with
column vectors, the function~$\sigma$ can also be applied to a column vector.

\begin{definition}[Pascal-like matrices]\label{def:Msigma}
  Let $d$ be a non negative integer and let $n=2^d$.  Let
  $m_0,\ldots,m_{n-1}$ be  integers such that $m_{n-1} = 0$
  and $m_{i+1} \le m_i \le m_{i+1}+1$ for each integer $0 \le i < n$.  Let
  $C_0,\ldots,C_{n-1}$ be the columns of $M_d$, that is, $M_d=
  (C_0,\ldots,C_{n-1})$. Define
  \begin{displaymath}
  M_d^{m_0,\ldots,m_{n-1}}= \bigl(\sigma^{m_0}(C_0),\ldots,\sigma^{m_{n-1}}(C_{n-1})\bigr).
  \end{displaymath}
\end{definition}

The following are the eight possible matrices $M_d^{m_0,\ldots,m_{n-1}}$ for $d
= 2$ and $n = 2^2$.

\begin{displaymath}
  \begin{array}{cccc} 
    M_2^{0,0,0,0} & M_2^{1,0,0,0} & M_2^{1,1,0,0} & M_2^{2,1,0,0} \\
    \begin{pmatrix}
      1 & 1 & 1 & 1 \\
      0 & 1 & 0 & 1 \\
      0 & 0 & 1 & 1 \\
      0 & 0 & 0 & 1
    \end{pmatrix}
    & 
    \begin{pmatrix}
      0 & 1 & 1 & 1 \\
      1 & 1 & 0 & 1 \\
      0 & 0 & 1 & 1 \\
      0 & 0 & 0 & 1
    \end{pmatrix} 
    & 
    \begin{pmatrix}
      0 & 0 & 1 & 1 \\
      1 & 1 & 0 & 1 \\
      0 & 1 & 1 & 1 \\
      0 & 0 & 0 & 1
    \end{pmatrix}
    & 
    \begin{pmatrix}
      0 & 0 & 1 & 1 \\
      0 & 1 & 0 & 1 \\
      1 & 1 & 1 & 1 \\
      0 & 0 & 0 & 1
    \end{pmatrix} \\[10mm]
    M_2^{1,1,1,0} & M_2^{2,1,1,0} & M_2^{2,2,1,0} & M_2^{3,2,1,0} \\
    \begin{pmatrix}
      0 & 0 & 0 & 1 \\
      1 & 1 & 1 & 1 \\
      0 & 1 & 0 & 1 \\
      0 & 0 & 1 & 1
    \end{pmatrix}
    & 
    \begin{pmatrix}
      0 & 0 & 0 & 1 \\
      0 & 1 & 1 & 1 \\
      1 & 1 & 0 & 1 \\
      0 & 0 & 1 & 1
    \end{pmatrix} 
    & 
    \begin{pmatrix}
      0 & 0 & 0 & 1 \\
      0 & 0 & 1 & 1 \\
      1 & 1 & 0 & 1 \\
      0 & 1 & 1 & 1
    \end{pmatrix}
    & 
    \begin{pmatrix}
      0 & 0 & 0 & 1 \\
      0 & 0 & 1 & 1 \\
      0 & 1 & 0 & 1 \\
      1 & 1 & 1 & 1
    \end{pmatrix}
  \end{array}
\end{displaymath}
\medskip

The matrix $M_d^{0,\ldots, 0}$ is exactly the matrix $M_d$ of
Definition~\ref{def:Md}.  Not only $M_d^{0,\ldots, 0}$ but all the matrices
$M_d^{m_0,\ldots,m_{n-1}}$ of Definition~\ref{def:Msigma} have the property
that all the square submatrices on the top and on the right are invertible.

A proof of this result appears 
in~\cite[Lemmas 4 and 5]{BecherCarton2019}.
We include them below.

\begin{lemma}[\protect{\cite[Lemma 4]{BecherCarton2019}}] \label{lem:invert-right-sm}
  Let $d$ be a non negative integer and let $n=2^d$.  Let matrix $M$ be a
  one of $M_d^{m_0,\ldots,m_{n-1}}$.  Let $\ell$ and~$k$ be two integers
  such that $0 \le \ell < \ell + k \leq n$.  The submatrix given by the $k$
  rows $\ell,\ell+1,\ldots,\ell+k-1$ and the last $k$ columns
  $n-k,\ldots,n-1$ of~$M$ is invertible.
\end{lemma}

\begin{proof}
  Note that for $k = n$ and $\ell = 0$, the submatrix in the statement of
  the lemma, is the whole matrix $M_d^{m_0,\ldots,m_{n-1}}$, which is
  clearly invertible.  By Lemma~\ref{lem:pascal}, each entry~$M_{i,j}$ for
  $0< i <n$ and $0\leq j < n-1$ of the matrix~$M$ satisfies either $M_{i,j}
  = M_{i-1,j} \oplus M_{i,j+1}$ if $m_j = m_{j+1}$ (the column~$C_j$ has
  been rotated as much as the column~$C_{j+1}$) or $M_{i,j} = M_{i-1,j}
  \oplus M_{i-1,j+1}$ if $m_j = m_{j+1} + 1$ (the column~$C_j$ has been
  rotated once more than the column~$C_{j+1}$).

  Let $P$ be the submatrix in the statement of the lemma:
      \begin{center}
    \begin{tikzpicture}
      \node at (-0.65,0) {$M = $};
      \node at (0,0) {$\left(\vphantom{\rule{0cm}{1cm}} \right.$};
      \node at (2,0) {$\left.\vphantom{\rule{0cm}{1cm}} \right)$};
      \node at (1.55,0.1) {$P$};
      \draw (1.15,-0.3) rectangle (1.95,0.5);
      \draw [<->] (1.15,-0.4) -- (1.95,-0.4);
      \node at (1.55,-0.6) {$k$};
      \draw [<->] (0,-0.4) -- (1.15,-0.4);
      \node at (0.6,-0.6) {$n-k$};
      \draw [<->] (1.05,0.5) -- (1.05,-0.3);
      \node at (0.9,0.1) {$k$};
      \draw [<->] (1.05,0.8) -- (1.05,0.5);
      \node at (0.9,0.65) {$\ell$};
      \draw [<->] (2.15,0.9) -- (2.15,-0.9);
      \node at (2.45,0) {$n$};
      \draw [<->] (0,-0.95) -- (2,-0.95);
      \node at (1,-1.2) {$n$};
    \end{tikzpicture}
    \end{center}
  To prove that~$P$ is invertible we apply
  transformations to make it triangular.  
  Note that all entries of the last column are~$1$.  
  The first transformation applied to~$P$ is as follows.
  The row $L_0$ is left unchanged and the row~$L_i$ for $1 \le i < k$ is
  replaced by $L_i \oplus L_{i-1}$.  
  All entries of the last column but its top most one become zero.  Furthermore, each entry is $P_{i,j}$ is 
  replaced by either $P_{i,j+1}$ or~$P_{i-1,j+1}$ depending on the value
  $m_j - m_{j+1}$.  
 Note also that the new values of the entries still
  satisfy either $P_{i,j} = P_{i-1,j} \oplus P_{i,j+1}$ or
  $P_{i,j} = P_{i-1,j} \oplus P_{i-1,j+1}$ depending on the value
  $m_j - m_{j+1}$.  
  
  The second transformation applied to~$P$ is as follows.
  The rows $L_0$ and~$L_1$ are left unchanged and each row~$L_i$ for
  $2 \le i < k$ is replaced by $L_i \oplus L_{i-1}$.  All entries of the
  second to last columns but its two topmost ones are now zero.  
  At step~$t$
  for $0 \le t < k$, rows $L_0,\ldots,L_t$ are left unchanged and each
  row~$L_i$ for $t+1 \le i < k$ is replaced by $L_i \oplus L_{i-1}$.
  After applying all these transformations for $0 \le t < k$, each
  entry~$P_{i,j}$ for $i+j = k-1$ satisfies $P_{i,j} = 1$ and each
  entry~$P_{i,j}$ for $i+j > k-1$ satisfies $P_{i,j} = 0$.  It follows that
  the determinant of~$P$ is~$1$ and that the matrix~$P$ is invertible.
\end{proof}

Let $n = 2^d$ for some $d ⩾ 1$ and let $M$ be one
matrix~$M_d^{m_0,\ldots,m_{n-1}}$.  We introduce the notions of
\emph{upper} and \emph{lower borders} of such a
matrix~$M_d^{m_0,\ldots,m_{n-1}}$.  An entry~$M_{i,j}$ for $0 ⩽ i,j < n$ is
said to be in the \emph{upper border} (respectively \emph{lower border})
of~$M$ if $M_{i,j} = 1$ and $M_{k,j} = 0$ for all $k=0,\ldots, i-1$
(respectively for all $k = i+1,\ldots,n-1$).  For example, the upper border
of the matrix~$M^{0,…,0}_d=M_d$ is the first row and its lower border is
the main diagonal.  The upper and lower borders in column~$i$ lie in lines
$m_i$ and $m_i+i$ respectively.

\begin{figure}[htbp]
\begin{displaymath}
  M_3^{3,3,2,1,1,1,0,0} = 
  \begin{pmatrix}
     0 & 0 & 0 & 0 & 0 & 0 & \underline{\mathbf{1}} & \mathbf{1} \\
     0 & 0 & 0 & \underline{\mathbf{1}} & \mathbf{1} & \mathbf{1} & 0 & 1 \\
     0 & 0 & \underline{\mathbf{1}} & 1 & 0 & 1 & 1 & 1 \\
     \underline{\mathbf{1}} & \mathbf{1} & 0 & 1 & 0 & 0 & 0 & 1 \\
     0 & \underline{\mathbf{1}} & \mathbf{1} & \mathbf{1} & 0 & 0 & 1 & 1 \\
     0 & 0 & 0 & 0 & \underline{\mathbf{1}} & 1 & 0 & 1 \\
     0 & 0 & 0 & 0 & 0 & \underline{\mathbf{1}} & \mathbf{1} & 1 \\
     0 & 0 & 0 & 0 & 0 & 0 & 0 & \underline{\mathbf{1}} 
  \end{pmatrix}
\end{displaymath}
\caption{Upper and lower borders of  $M_3^{3,3,2,1,1,1,0,0}$
         are shown in boldface.}
\label{fig:upperlower}
\end{figure}

\begin{displaymath}
  \begin{array}{r|c|c|c|c|c|c|c|c}
    \text{Column $i$}   & 0 & 1 & 2 & 3 & 4 & 5 & 6 & 7 \\ \hline
    \text{Upper border $m_i$} & 3 & 3 & 2 & 1 & 1 & 1 & 0 & 0 \\
    \text{Lower border $m_i+i$} & 3 & 4 & 4 & 4 & 5 & 6 & 6 & 7 
  \end{array}
\end{displaymath}
\medskip

Both borders of matrix~$M_d^{m_0,\ldots,m_{n-1}}$ start in the unique
entry~$1$ of the first column.  The upper border ends in the top most entry
of the last column and the lower border ends in the bottom most entry of
the last column.  The upper border is only made of either East or North-East
steps and the lower border is only made of either East or South-East steps.  The
upper border contains an East step from column~$C_j$ to column~$C_{j+1}$ if
$m_j=m_{j+1}$ and contains a North-East step if $m_j = m_{j+1}+ 1$.
Furthermore, whenever the upper border uses an East (respectively
North-East) step to go from one columns to its right neighbour, the lower
border uses a South-East (respectively East) step.  This is because the
distance from the upper border to the lower border in column~$i$ is
exactly~$i$.  This allows us to define a function~$τ$ from $\{0,…,n-1\}$ to
$\{0,…,n-1\}$ as follows.

\begin{displaymath}
  τ(i) =
  \begin{cases}
    m_i   & \text{if either $i = 0$ or $m_{i-1} = m_i + 1$} \\
    m_i+i & \text{otherwise, that is, $i > 0$ and $m_{i-1} = m_i$}
  \end{cases}
\end{displaymath}
\medskip

The value of $τ(i)$ is the line index of either the upper or the lower
border in column~$i$.  It follows from the definition of the function~$τ$
that $M_{τ(i),i} = 1$ and $M_{τ(i),j} = 0$ for each $0
⩽ j < i$.  The values of the function~$τ$ for the matrix of
Figure~\ref{fig:upperlower} are given below.  In
Figure~\ref{fig:upperlower}, the $1$s of the borders corresponding to values
of~$τ$ are underlined.  Note that there is exactly a single underlined~$1$
in each line and in each column.

\begin{displaymath}
  \begin{array}{c|c|c|c|c|c|c|c|c}
          i    & 0 & 1 & 2 & 3 & 4 & 5 & 6 & 7 \\ \hline
          τ(i) & 3 & 4 & 2 & 1 & 5 & 6 & 0 & 7 
  \end{array}
\end{displaymath}
\medskip

The function~$τ$ is onto and thus bijective because each leftmost
occurrence of~$1$ in each line belongs to either the upper or the lower
border.  It follows that each $j$ in~$\{0,…,n-1\}$ is equal to~$τ(i)$ where
$i$ is the column of the leftmost~$1$ in line~$j$.

Due to the symmetry in the matrix $M^{0,\ldots,0}_d=M_d$,
Lemma~\ref{lem:invert-right-sm} applies also to the submatrices of
$M^{0,\ldots,0}_d$ obtained by selecting the first row.  Since this
symmetry is lost in the other matrices $M_d^{m_0,\ldots,m_{n-1}}$, we need 
to consider the rotations made to the columns
in~$M^{0,\ldots,0}_d$ to obtain $M_d^{m_0,\ldots,m_{n-1}}$.

\begin{lemma}[\protect{\cite[Lemma5]{BecherCarton2019}}] \label{lem:invert-top-sm}
  Let $d$ be a non negative integer and let $n=2^d$.
  Let matrix  $M$ be  one of $M_d^{m_0,\ldots,m_{n-1}}$.  
  Let $k$ be an
  integer such that $1 \le k \leq n$.  
  The $k\times k$-submatrix of $M$ 
  given by  $k$ consecutive rows and $k$ consecutive columns 
  with its top right entry
  on the upper border of~$M$ is  invertible.
\end{lemma}

\begin{proof}
   Let $P$ be the submatrix of $M$ in the statement of the lemma:
   \begin{center}
      \begin{tikzpicture}
      \node at (-0.65,0) {$M = $};
      \node at (0,0) {$\left(\vphantom{\rule{0cm}{1cm}} \right.$};
      \node at (2,0) {$\left.\vphantom{\rule{0cm}{1cm}} \right)$};
      \node at (1.05,0.1) {$P$};
      \draw (0.65,-0.3) rectangle (1.45,0.5);
      \draw [<->] (0.65,-0.4) -- (1.45,-0.4);
      \node at (1.05,-0.6) {$k$};
      \draw [<->] (0.55,0.5) -- (0.55,-0.3);
      \node at (0.4,0.1) {$k$};
      \node at (1.39,0.4) {$\scriptscriptstyle 1$};
      \node at (1.39,0.59) {$\scriptscriptstyle 0$};
      \node at (1.39,0.74) {$\scriptscriptstyle 0$};
      \draw [<->] (2.15,0.9) -- (2.15,-0.9);
      \node at (2.45,0) {$n$};
      \draw [<->] (0,-0.95) -- (2,-0.95);
      \node at (1,-1.2) {$n$};
    \end{tikzpicture}
  \end{center}
  We apply transformations to the submatrix~$P$ to put it in a nice form
  such that the determinant is easy to compute.  To fix notation, suppose
  that the submatrix~$P$ is obtained by selecting rows
  $L_{r},\ldots,L_{r+k-1}$ and columns $C_{s},\ldots,C_{s+k-1}$.  The
  hypothesis is that the entry  $M_{r,s+k-1}$ is in the upper border of~$M$.  Note
  that the upper borders of $M$ and~$P$ coincide inside~$P$.  We denote by
  $j_0,\ldots, j_t$ the indices of the columns of~$P$ in $0,\ldots,k-1$
  originally defined by a North-East step of the upper border.  This means
  that $j_0,\ldots,j_t$ is the sequence of indices~$j$ such that 
  $m_{s+j-1} = m_{s+j} + 1$.  
By convention, we set $j_0 = 0$, that
  is, the index of the first column of~$P$.
 
  The first transformation applied to the matrix~$P$ is the following. 
  The columns $C_0,\ldots,C_{j_t-1}$ and~$C_{k-1}$ are left unchanged and each
  column~$C_j$ for $j_{t} \le j < k-1$ is replaced by $C_j \oplus C_{j+1}$.
  All entries of the first row but its right most one become~zero.
  Furthermore, each entry $P_{i+1,j}$ for $j_{t} \le j < k-1$ is replaced by
  $P_{i,j}$.   
  The second transformation applied to the matrix~$P$ is the  following.  
  The columns $C_0,\ldots,C_{j_{t-1}-1}$ and~$C_{k-2},C_{k-1}$ are
  left unchanged and each column~$C_j$ for 
  $j_{t-1} \le j < k-2$ is
  replaced by $C_j \oplus C_{j+1}$.  
    The first row remains unchanged and
  all entries of the second row but the last two become~$0$.  
  We apply in total $t+1$
  transformations like this one using successively $j_t,j_{t-1},\ldots,j_0$.
Then $k-t$ further steps are made, obtaining that for each row $i$,  all entries   but the last $i$ become~$0$.  

  After applying all these transformations, each entry~$P_{i,j}$ for
  $i+j = k-1$ satisfies $P_{i,j} = 1$ and each entry~$P_{i,j}$ for
  $i+j < k-1$ satisfies $P_{i,j} = 0$.  It follows that the determinant
  of~$P$ is~$1$ and that the matrix~$P$ is invertible.
\end{proof}

We  define the \emph{affine toroidal arrays}
 by considering  the family of Pascal-like  matrices.

\begin{definition}[Affine toroidal array]\label{def:affine}
Let $d$ be a non-negative integer and let  $n = 2^d$.  
Let $M$ be any   Pascal-like matrix of size $n\times n$ (Definition~\ref{def:Msigma}).
Let $N_0,\ldots,N_{2^{n^2}-1}$ be the enumeration 
of all~$n\times n$ matrices over~$\FT$
(Definition~\ref{def:enum})
and let $Z$ be any $n\times n$ matrix over~$\FT$.

An  array $A$  of size
  $ n2^{n^2/2} × n2^{n^2/2}$
is   $(n,n,n,n)$-affine 
if    for each integer~$k$ such that $0 ⩽ k < 2^{n^2}$,
the matrix $M N_k \oplus Z$ is placed in~$A$ with  its upper left
corner cell  at position $(\odd(k)n, \even(k)n)$. 
\end{definition}

Since  each  matrix $M=M_d^{m_0,\ldots,m_{n-1}}$ is
 invertible, every  matrix~$Z$ is equal to~$MZ'$ for some matrix~$Z'$.

For an array $Z$ we write $(Z)^{(n\times n)}$ to denote the array  given by  
 $n^2$ copies of $Z$, $n$ rows and $n$ columns.
The next result states that any nested
perfect array can be transformed into another one of the same size but having the matrix $0$ in the upper corner.

In what follows we use the operation $\oplus$ on subarrays denoting the usual sum on matrices of elements in~$\mathbb{F}_2$.

\begin{lemma} \label{lem:xor-npn}
Let $d$ be a non-negative integer. 
Fix  $n = 2^d$ .  
  Let $A$ be an array of size   $ n2^{n^2/2} × n2^{n^2/2}$
 and let $Z$ be an array of size $n\times n$.
  Then  $A$ is a nested $(n,n,n,n)$-perfect toroidal array if and only if
  $A \oplus (Z)^{{2^{n^2/2}}\times{2^{n^2/2}}}$ is a nested $(n,n,n,n)$-perfect toroidal array.
\end{lemma}

\begin{proof}
Let $A'=A\oplus(Z)^{2^{n^2/2}\times 2^{n^2/2}}$.  
Let $\ell$ be an integer such
  that $1 \le \ell \leq n$ 
  and  let $L$ be an aligned  subarray~of $A$ of
  size~$n2^{\ell}\times n 2^{\ell} $ starting at a position congruent to~$(0,0)$.  
The corresponding subarray $L'$ of $A'$ at the same position is $L' = L \oplus (Z)^{(2^\ell \times 2^\ell)}$. 
  
  First suppose $A$ is nested $(n,n,n,n)$-perfect. Then,~$L$ is a nested $(\ell,\ell,n,n)$-perfect array. 
 Let $i,j$ be such that $0\le i,j < n $
 and  $T$ be the subarray of~$(Z)^{(2\times 2)}$ of size~$\ell\times\ell$ starting at
  position~$(i,j)$.
  Let $U'$ be any array of
 size~$\ell\times\ell$. 
  Then, the array 
  $U = U' \oplus T$ has exactly
  one occurrence in the array~$L$ at some position~$(i',j')$ congruent to~$(i,j)$  modulo~$(n, n)$. 
  It follows that $U'=U\oplus T$ has an occurrence at the
  same position~$(i',j')$ in~$L'$.  
  Since each matrix~$U$ has such an occurrence
  for each possible~$(i,j)$.Since $L'$ has size $n2^{\ell}\times n2^{\ell}$, $L'$ it is a nested
  $(\ell,\ell,n,n)$-perfect array.
  
If  $A$ is not nested $(n,n,n,n)$-perfect, there is a witness  $\ell$,
 $1 \le \ell \leq n$ 
 and a subarray $L$ of size $n2^\ell\times n2^\ell$ that occurs twice at positions in  the same congruence class. With a similar argument as in the previous case it is easy to check that there is a  subarray $L'$ of $A'$ with   two occurrences in $A'$ in the same  congruence class.
 \end{proof}

We can now prove that the affine arrays are  nested perfect arrays.

\begin{proposition} \label{pro:aff2nes}
  Let $d$ be a non negative integer and let   $n = 2^d$.
Every $(n,n,n,n)$-affine array is a nested $(n,n,n,n)$-perfect array.
\end{proposition}

\begin{proof}
  By Lemma~\ref{lem:xor-npn}, it may be assumed that the array~$Z$ in the
  definition of affine array is the zero matrix.  Let $M$ be one of the
  matrices $M_d^{m_0,\ldots,m_{n-1}}$.  Suppose $A$ is the
  $(n,n,n,n)$-affine array defined by $M$ and the zero matrix.  The proof
  of the Proposition follows that of Proposition~\ref{prop:main}, but now
  considering the matrix $M=M_d^{m_0,\ldots,m_{n-1}}$ instead of the Pascal
  matrix $M_d^{0,\ldots,0}$ Let $k$ be an integer such that $1 \le k \le
  n$.  Let $B$ be an aligned subarray of~$A$ of sizes $n2^{kn/2} ×
  n2^{kn/2}$.  The coordinates of the upper left corner of $B$ are of the
  form $pn2^{kn/2}$ and $qn2^{kn/2}$ for two integers $p$ and~$q$ such
  that $0 ⩽ p,q < 2^{(n-k)n/2}$.  This means that the subarray~$B$ is tiled
  by the matrices $M N_{\join{ℓ}{m}}$ for $ℓ$ and $t$ satisfying 
  \[
  p2^{kn/2}⩽ ℓ < (p+1)2^{kn/2}
 \text{ and }
  q2^{kn/2}⩽ t < (q+1)2^{kn/2}.
  \] 
  The binary
  expansions of all integers $ℓ$ satisfying \[
  p2^{kn/2} ⩽ ℓ < (p+1)2^{kn/2}
  \]
  start with the same $2^{n(n-k)/2}$ binary digits and the same hold for
  all integers $t$ satisfying 
  \[
  q2^{kn/2}⩽ t < (q+1)2^{kn/2}.
  \]
    This implies
  that $\join{ℓ}{m}$ start with the same $n(k-n)$ digits.  Since the first
  digits of $\join{ℓ}{m}$ are put in the first rows of $N_{\join{ℓ}{m}}$
  which have length~$n$, all matrices $N_{\join{ℓ}{m}}$ for $ℓ$ and $t$ in
  their respective intervals have the same first $n - k$ rows.

  Let $(i,j)$ be a pair of integers such that $0 ⩽ i,j < n$ and let $P$ be
  an array of sizes $k × n$.  We claim that $P$ has exactly one occurrence
  in~$B$ which is congruent to $(i,j)$ modulo~$(n,n)$.
  In order to prove
  it, we show that $P$ has a single such occurrence exactly when a certain
  system of linear equations has a solution.  Furthermore, this solution of
  the system provides the matrix~$N_{\join{ℓ}{m}}$ and thus the integers
  $ℓ$ and~$m$ which, in turn, give the position of the occurrence of~$P$ in
  the subarray~$B$.  An occurrence~$P$ can overlap at most four matrices
  tiling the subarray~$B$.  
  
  Suppose that the upper left corner of the
  occurrence of~$P$ lies in some matrix 
  $M N_{\join{ℓ}{m}}$
  where the
  integers $ℓ$ and~$m$ such that $p2^{kn/2} ⩽ ℓ < (p+1)2^{kn/2}$ and
  $q2^{kn/2} ⩽ m < (q+1)2^{kn/2}$.  The matrix on the right
  of~
$M N_{\join{ℓ}{m}}$ 
and the matrix below it are respectively
$M N_{\join{(ℓ+1)}{m}}$ and $M N_{\join{ℓ}{(m+1)}}$
where $ℓ+1$ and
  $m+1$ must be understood modulo $2^{kn/2}$ in order to remain in the
  right intervals.  Let $P_1$, $P_2$, $P_3$ and~$P_4$ be the parts of~$P$
  that overlap respectively the matrices
$M N_{\join{ℓ}{m}}$,
  $M N_{\join{ℓ}{(m+1)}}$, $M N_{\join{(ℓ+1)}{m}}$ and
  $M N_{\join{(ℓ+1)}{(m+1)}}$.
  They are pictured in
  Figure~\ref{fig:occurrence}.
  
  If $j = 0$, the parts $P_2$ and~$P_4$ of the occurrence are empty.  This is a degenerate case, so we only treat the case $j ⩾ 1$.
  Consider two main cases depending on whether $i + k ⩽ n$ or not.

  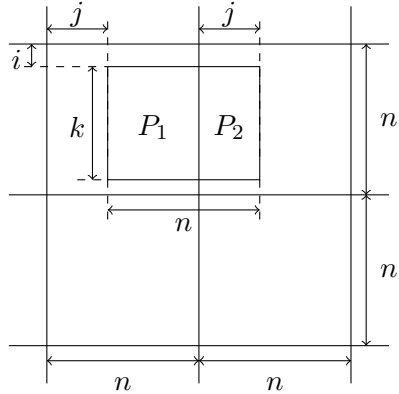
\begin{figure}[htbp]
    \begin{center}
      \begin{tikzpicture}
        \draw (-0.5,4) -- (4.5,4);
        \draw (-0.5,2) -- (4.5,2);
        \draw (-0.5,0) -- (4.5,0);
        \draw (0,-0.5) -- (0,4.5);
        \draw (2,-0.5) -- (2,4.5);
        \draw (4,-0.5) -- (4,4.5);
        \draw (0.8,3.7) -- (2.8,3.7) -- (2.8,2.2) -- (0.8,2.2) -- cycle;
        \draw[dashed] (-0.3,3.7) -- (0.8,3.7);
        \draw[<->] (-0.2,3.7) -- (-0.2,4);
        \node at (-0.4,3.8) {$i$};
        \draw[dashed] (0.8,4.3)  -- (0.8,2.5);
        \draw[<->] (0,4.2) -- (0.8,4.2);
        \node at (0.4,4.4) {$j$};
        \draw[dashed] (2.8,4.3)  -- (2.8,2.5);
        \draw[<->] (2,4.2) -- (2.8,4.2);
        \node at (2.4,4.4) {$j$};
        \draw[dashed] (0.4,2.2) -- (0.8,2.2);
        \draw[<->] (0.6,3.7) -- (0.6,2.2);
        \node at (0.4,2.9) {$k$};
        \draw[dashed] (0.8,2.2)  -- (0.8,1.6);
        \draw[dashed] (2.8,2.2)  -- (2.8,1.6);
        \draw[<->] (0.8,1.8) -- (2.8,1.8);
        \node at (1.8,1.6) {$n$};

        \draw[<->] (0,-0.2) -- (2,-0.2);
        \node at (1,-0.5) {$n$};
        \draw[<->] (2,-0.2) -- (4,-0.2);
        \node at (3,-0.5) {$n$};
        \draw[<->] (4.2,0) -- (4.2,2);
        \node at (4.5,1) {$n$};
        \draw[<->] (4.2,2) -- (4.2,4);
        \node at (4.5,3) {$n$};
        \node at (1.4,2.9) {$P_1$};
        \node at (2.4,2.9) {$P_2$};
      \end{tikzpicture}
    \end{center}
    \caption{An occurrence of array $B$ in $A$ with $i + k ⩽ n$.}
    \label{fig:occurrence2}
  \end{figure}
  
Suppose that $i + k ⩽ n$.  
Then,
the parts $P_3$ and~$P_4$ do not exist and the occurrence
  of~$P$ is reduced to $P_1$ and~$P_2$.  Consider a column of the
  occurrence in~$P_1$, that is, a column of the matrix 
  $M N_{\join{ℓ}{m}}$
  with index~$s$ greater than~$j$ from line $i+1$ to
  line~$i+k$.  This column is obtained by multiplying the lines from $i+1$
  to~$i+k$ of~$M$
  with the column of index~$s$ of~$N_{\join{ℓ}{m}}$.
  Note that the first $n-k$ entries of this latter are known and can be
  considered as constant. The $k$ remaining entries of the column~$s$
  of~$N_{\join{ℓ}{m}}$ are thus the solution of the system $y = Mx$ where
  $y$ is the column of~$P_1$, $M$ is the $k × k$ matrix made of lines from
  $i$ to~$i+k-1$ and columns $n-k$ to~$n-1$ of~$M$
and $x$ are the $k$
  entries of~$N_{\join{ℓ}{m}}$.  By Lemma~\ref{lem:invert-right-sm}, the
  matrix~$M$ is invertible and there is then a unique solution to the
  system.  This means that the $k$ entries of the column of index~$s$
  of~$N_{\join{ℓ}{m}}$ can be found.  An similar reasoning with a column
  of~$P_2$ allows us to find a column $s$ with $s ⩽ j$
  of~$N_{\join{ℓ}{(m+1)}}$ and thus of~$N_{\join{ℓ}{m}}$.

  \begin{figure}[htbp]
    \begin{center}
      \begin{tikzpicture}
        \draw (-0.5,4) -- (4.5,4);
        \draw (-0.5,2) -- (4.5,2);
        \draw (-0.5,0) -- (4.5,0);
        \draw (0,-0.5) -- (0,4.5);
        \draw (2,-0.5) -- (2,4.5);
        \draw (4,-0.5) -- (4,4.5);
        \draw (0.8,2.5) -- (2.8,2.5) -- (2.8,1) -- (0.8,1) -- cycle;
        \draw[dashed] (-0.3,2.5) -- (0.8,2.5);
        \draw[<->] (-0.2,2.5) -- (-0.2,4);
        \node at (-0.4,3.2) {$i$};
        \draw[dashed] (-0.3,0.5) -- (0.8,0.5);
        \draw[<->] (-0.2,0.5) -- (-0.2,2);
        \node at (-0.4,1.2) {$i$};
        \draw[dashed] (0.8,4.3)  -- (0.8,2.5);
        \draw[<->] (0,4.2) -- (0.8,4.2);
        \node at (0.4,4.4) {$j$};
        \draw[dashed] (2.8,4.3)  -- (2.8,2.5);
        \draw[<->] (2,4.2) -- (2.8,4.2);
        \node at (2.4,4.4) {$j$};
        \draw[dashed] (0.4,1) -- (0.8,1);
        \draw[<->] (0.6,1) -- (0.6,2.5);
        \node at (0.4,1.7) {$k$};
        \draw[<->] (0.6,0.5) -- (0.6,1);
        \node at (1.1,0.75) {$n{-}k$};
        \draw[dashed] (2.8,2.5) -- (3.2,2.5);
        \draw[<->] (3,2) -- (3,2.5);
        \node at (3.45,2.25) {$n{-}i$};
        \draw[<->] (0.8,3.3) -- (2.8,3.3);
        \node at (1.8,3.5) {$n$};

        \draw[<->] (0,-0.2) -- (2,-0.2);
        \node at (1,-0.5) {$n$};
        \draw[<->] (2,-0.2) -- (4,-0.2);
        \node at (3,-0.5) {$n$};
        \draw[<->] (4.2,0) -- (4.2,2);
        \node at (4.5,1) {$n$};
        \draw[<->] (4.2,2) -- (4.2,4);
        \node at (4.5,3) {$n$};
        \node at (1.4,2.2) {$P_1$};
        \node at (2.4,2.2) {$P_2$};
        \node at (1.4,1.5) {$P_3$};
        \node at (2.4,1.5) {$P_4$};
      \end{tikzpicture}
    \end{center}
    \caption{An occurrence of array $B$ in $A$ with $i + k > n$.}
    \label{fig:occurrence}
  \end{figure}
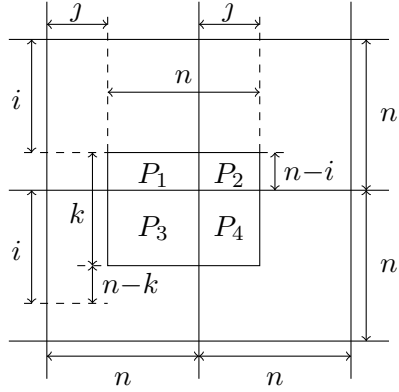
  
  Now we suppose that $i + k > n$ and we make more explicit how the
  matrix~$N_{\join{ℓ}{m}}$ can be computed from the occurrence of~$P$
  modulo $(i,j)$.  
  Let us recall that the $n-k$ top lines
  of~$N_{\join{ℓ}{m}}$, that is, lines $0,…,n-k-1$ are fixed by the
  subarray~$B$.  Let $n-k,…, n-1$ the indices of the lines
  of~$N_{\join{ℓ}{m}}$ which are still unknown.  The computation of these
  $k$ remaining lines is carried out in two phases.  The second and first
  phases respectively compute the lines $n-k,…, r-1$ and $r, …, n-1$ where
  the cutting index~$r$ satisfying $n-k ⩽ r ⩽ n-1$ is defined as follows.
  The integer~$r$ is the least integer such that all the integers
  $τ(r), … , τ(n-1)$ belong to $\{0, …, i+k-n-1\} ∪ \{i, …, n-1\}$.  Note
  that $\{0, …, i+k-n-1\}$ are the indices of the lines crossing $P_3P_4$
  while $\{i, …, n-1\}$ are the indices of the lines crossing $P_1P_2$.  If
  $r = n-k$, all $k$ lines $n-k,…,n-1$ are computed by the first phase and
  the second phase is void.

  We describe more precisely the first phase. Let $s$ be an integer
  satisfying $r ⩽ s ⩽ n-1$ and let us suppose that lines $s+1,…,n-1$ are
  already known and that line~$s$ is still unknown.  The entries of
  line~$s$ are computed from the rightmost one of index~$n-1$ to the
  leftmost one of index~$0$.  Let us consider the line $τ(s)$ in the
  matrix~$M$.  By definition of~$τ$, the entry $M_{τ(s),s}$ is equal to~$1$
  and entries $M_{τ(s),s'}$ for $s' < s$ are equal to~$0$.  We consider two
  cases depending on whether $τ(s)$ belongs to $\{0, …, i+k-n-1\}$ or to
  $\{i, …, n-1\}$.  

Assume $τ(s)$ belongs to
  $\{i, …, n-1\}$.  Let $t$ be an integer such that $j ⩽ t ⩽ n-1$.  The
  multiplication of the line~$τ(s)$ in~$M$ and the column~$t$ of
  $N_{\join{ℓ}{m}}$ gives the entry $(τ(s), t)$ in~$P_1$, that is, the
  entry $(τ(s)-i, t-j)$ in~$P$.  The properties of the line $τ(s)$ in~$M$
  and the fact that lines below line~$r$ in~$N_{\join{ℓ}{m}}$ are already
  known allow us to compute the entry $(s,t)$ in~$N_{\join{ℓ}{m}}$.  Let
  $t$ be an integer such that $0 ⩽ t < j$.  The multiplication of the
  line~$τ(s)$ in~$M$ and the column~$t$ of $N_{\join{ℓ}{m+1}}$ gives the
  entry $(τ(s), t)$ in~$P_2$, that is, $(τ(s)-i, t+n-i)$ in~$P$.
  
  The properties of the line $τ(s)$ in~$M$ and the fact that lines below
  line~$s$ in~$N_{\join{ℓ}{m}}$ are already known allow us to find the
  entry $(s,t)$ in~$N_{\join{ℓ}{(m+1)}}$.  Since all entries $(s,t)$ for
  $j ⩽ t ⩽ n-1$ of $N_{\join{ℓ}{m}}$ and all entries $(s,t)$ for
  $0 ⩽ t < j$ of $N_{\join{ℓ}{(m+1)}}$ and all lines below line~$s$ in
  $N_{\join{ℓ}{m}}$ are known, line~$s$ of $N_{\join{ℓ}{m}}$ is known. 
  
  Now assume  $τ(s)$ belongs to $\{0, …, i+k-n-1\}$.  The same
  reasoning with matrices $N_{\join{(ℓ+1)}{m}}$ and
  $N_{\join{(ℓ+1)}{(m+1)}}$ and parts $P_3$ and $P_4$ of~$P$ allows us to
  compute line~$s$ of $N_{\join{(ℓ+1)}{m}}$ and thus line~$s$ of
  $N_{\join{ℓ}{m}}$.
  
  We finally describe the second phase.  Lines $0,…,n-k-1$ are fixed by the
  subarray~$B$ and lines $r,…,n-1$ have been computed by the first phase.
  Lines $n-k,…,r-1$ are still unknown.  We assume that $n-k < r$ since
  otherwise no line is unknown. It follows from the definition of~$r$ that
  the integer $τ(r-1)$ is then either $i+k-n$ or $i-1$.  Suppose 
  $τ(r-1) = i-1$, the other case is similar.  Consider the
  $(k+r-n) × (k+r-n)$ matrix~$M'$ obtained by selecting lines
  $i-r,…,i+k-n-1$ and columns $n-k,…, r-1$ from the matrix~$M$. The upper
  right entry of~$M'$ is the entry $(i-r, r-1)$ of~$M$.  Since $τ(r-1) =
  i-1$ and the distance between the upper and the lower borders in column
  $r-1$, is $r-1$, the entry $(i-r, r-1)$ lies on the upper border of~$M$.
  By Lemma~\ref{lem:invert-top-sm}, the matrix~$M'$ is invertible.  Note
  that selected lines of~$M'$ are still unused lines of~$P$ and that
  selected columns correspond to still unknown lines of~$N_{\join{ℓ}{m}}$.
  Invertibility of~$M'$ allows us to compute lines $n-k,…, r-1$
  of~$N_{\join{ℓ}{m}}$.  This completes the proof of the theorem.
\end{proof}

The next lemma computes the number of $(n,n,n,n)$-affine arrays.  

\begin{proposition} \label{pro:countaffine}
  Let $d$ be a non-negative integer and let $n=2^d$. Then,, there are
  $2^{n^2+n-1}$ $(n,n,n,n)$-affine arrays.
\end{proposition}

\begin{proof}
  There are exactly $2^{n-1}$ matrices $M_d^{m_0,\ldots,m_{m-1}}$.  Indeed,
  the sequence $m_0,\ldots,m_{n-1}$ is fully determined by the sequence
  $m_0-m_1,\ldots,m_{n-2}-m_{n-1}$ of $n-1$ differences which take their
  value in~$\{0,1\}$.  There are also $2^{n^2}$ possible values for the
  matrix~$Z$ in~$\FT^{n\times n}$.  This proves that the number of
  $(n,n,n,n)$-affine arrays is at most~$2^{n^2+n-1}$.
  
  It remains to show that two $(n,n,n,n)$-affine array obtained for two
  different pairs $(M,Z)$ and $(M',Z')$ are indeed different.  Let
  $N_0,\ldots,N_{2^{n^2}-1}$ be the enumeration of all $n\times n$ matrices
  over~$\FT$.  Let $M$ and $M'$ be two matrices of the form
  $M_d^{m_0,\ldots,m_{n-1}}$.  Let $Z$ and~$Z'$ be two  $n\times n$
  matrices~$\FT$.  Let $U_i = N_i \oplus Z$ and $U'_i = N_i \oplus Z'$ for
  $i=0, ...,  2^n-1$.  Let $W$ and $W'$ be the two placements is defined as
  follows: for each integer~$i$ such that $0 ⩽ i ⩽ 2^{2^n}-1$, the matrix
  $M U_k$ is placed in~$W$ in such a way that its upper left corner cell is
  at position $(\odd(i)n, \even(i)n)$.  Similarly for $W'$ using $U'_i$
  instead of $U_i$.  We claim that if $W = W'$, then $M = M'$ and~$Z = Z'$.
  We suppose that $W = W'$.  Since both matrices $M$ and~$M'$ are
  invertible by Lemma~\ref{lem:invert-right-sm}, $MU_i$ (respectively
  $M'U'_i$) is the zero vector if and only if $U_i$ (respectively $U'_i$)
  is the zero vector, that is, $Z = W_i$ (respectively $Z' = W_i$).  This implies that $Z = Z'$ and thus $U_i = U'_i$ for $i=0,..,  2^n-1$.
  Note that the matrix $U_i$ ranges over all possible $n\times n$ matrices.
  If $MU_i = M'U_i$ for all $i=0 , ...,  2^n-1$, then~$M = M'$.
\end{proof}

For $d$ a non negative integer and $n=2^d$ Definition~\ref{def:affine} gives a construction method of $(n,n,n,n)$-affine arrays. Proposition ~\ref{pro:countaffine} counts how many can be constructed and proves that they are all different.  This completes the proof of Theorem~\ref{thm:constructions}.

\section*{Acknowledgements.}
Both authors are members of SINFIN Laboratory Université de
Paris/CNRS-Universidad de Buenos Aires/CONICET.  This research is supported
by grant PICT 2018-2315, Agencia Nacional de Promoción Científica y
Tencológica de Argentina.
\medskip

\bibliographystyle{plain}
\bibliography{array}

\bigskip
\bigskip
\bigskip

{\small
\begin{minipage}{\textwidth}
\noindent
Ver\'onica Becher \\
Departamento de  Computaci\'on,
Facultad de Ciencias Exactas y Naturales \& ICC \\
Universidad de Buenos Aires \&  CONICET, Argentina \\
vbecher@dc.uba.ar
\bigskip\\
Olivier Carton \\
Institut de Recherche en Informatique Fondamentale \\
Universit\'e Paris Cité, France \\
Olivier.Carton@irif.fr
\end{minipage}
}

\end{document}